\documentclass[twocolumn]{autart}    

\usepackage{wrapfig}
\usepackage{natbib}
\usepackage{picins}
\usepackage{bm}
\usepackage{amsmath}
\usepackage{mathtools}
\usepackage{color}
\usepackage{amsfonts}
\usepackage{epsfig}
\usepackage{float}
\usepackage{mathrsfs}
\usepackage{caption}
\usepackage{subcaption}
\usepackage{longtable}
\usepackage{xtab}
\usepackage{url}
\DeclareMathOperator{\E}{\mathbb{E}}

\DeclareMathOperator*{\argmin}{arg\,min}
\newcommand{\nbf}{\noindent\textbf}
\newcommand{\nit}{\noindent\textit}

\newcommand{\R}{\mathbb R}

\newcommand{\F}{\mathcal F}
\newcommand{\Ss}{\mathbb S}

\newcommand{\Pb}{\mathbb P}
\usepackage{biolinum}

\DeclarePairedDelimiter\ceil{\lceil}{\rceil}

\usepackage{algorithm}
\usepackage{algpseudocode}

\usepackage{amssymb}

\makeatletter
\DeclareRobustCommand{\qed}{%
	\ifmmode 
	\else \leavevmode\unskip\penalty9999 \hbox{}\nobreak\hfill
	\fi
	\quad\hbox{\qedsymbol}}

\newcommand{\openbox}{\leavevmode
	\hbox to.77778em{%
		\hfil\vrule
		\vbox to.675em{\hrule width.6em\vfil\hrule}%
		\vrule\hfil}}

\newcommand{\qedsymbol}{\openbox}

\newenvironment{proof}[1][\proofname]{\par
	\normalfont
	\topsep6\p@\@plus6\p@ \trivlist
	\item[\hskip\labelsep\itshape
	#1.]\ignorespaces
}{%
	\qed\endtrivlist
}
\newcommand{\proofname}{Proof}
\makeatother

\newcommand{\red}[1]{\textcolor{black}{#1}}

\allowdisplaybreaks
\begin{document}

\begin{frontmatter}

\title{A Reliability-aware Multi-armed Bandit Approach to Learn and Select Users in Demand Response} 

\thanks[footnoteinfo]{The work was supported by NSF CAREER 1553407,NSF ECCS 1839632, AFOSR YIP, ONR YIP, and ARPA-E through the NODES program. }

\author[YL]{Yingying Li}\ead{yingyingli@g.harvard.edu},    
\author[QH]{Qinran Hu}\ead{qhu@seu.edu.cn},               
\author[YL]{Na Li}\ead{nali@seas.harvard.edu}  

\address[YL]{Harvard University, Cambridge, MA 02138, USA.}
\address[QH]{Southeast University, Nanjing, China}  


\begin{abstract}                          
	One challenge in the optimization and control of societal systems is to handle the unknown and uncertain user behavior. 
This paper focuses on residential demand response (DR) and proposes a closed-loop learning scheme to address these issues. In particular, we consider DR programs where an aggregator calls upon
residential users to change their demand so that the total load adjustment is close to a target value. 
To learn and select the right users,  we formulate the DR problem as a combinatorial multi-armed bandit (CMAB) problem with  a reliability objective. We propose a learning algorithm: CUCB-Avg (Combinatorial Upper Confidence Bound-Average), which utilizes both upper confidence bounds and sample averages to balance the tradeoff between exploration (learning) and exploitation (selecting). We consider both a fixed time-invariant target and time-varying targets, and show that 
CUCB-Avg achieves $O(\log T)$ and $O(\sqrt{T 
	\log(T)})$  regrets respectively. Finally,
we  numerically test our algorithms using synthetic and real  data, and demonstrate that our CUCB-Avg performs significantly better than the classic CUCB and also better than Thompson Sampling. 
\end{abstract}

\end{frontmatter}

\section{Introduction}

Unknown and uncertain user behavior is common in many sequential decision-making problems of societal systems, such as transportation, electricity grids, communication, crowd-sourcing, and resource allocation problems in general \citep{o2010residential,belleflamme2014crowdfunding,kuderer2015learning,li2017mechanism}. One key challenge caused by the unknown and uncertain user behavior is how to ensure reliability or reduce risks for the system.  This paper focuses on addressing this challenge for residential demand response (DR) in power systems.

 Residential DR refers to adjusting  power consumption of residential users, e.g. by changing the temperature setpoints of air conditioners,  to relieve  the supply-demand imbalances of the power system \citep{DRreportFERC2017,conEdison,o2010residential,thinkeco,psegcoolcustomer}. In most residential DR programs, customers can decide to respond to a DR signal or not, and the decisions are usually highly uncertain. Moreover, the pattern of the user behavior is not well understood by the DR aggregator. Such unknown and uncertain behavior may cause severe troubles for the system reliability: without enough knowledge of the user behavior,  the DR load adjustment  is likely to be very different from a target level, resulting in extra power imbalances and fluctuations. Therefore, it is critical for residential DR programs to learn the user behavior and ensure reliability  during the learning.


 Multi-armed bandit (MAB) emerges as a natural framework to learn the user  behavior \citep{auer2002finite,bubeck2012regret}.  In a simple setting, MAB considers $n$ independent arms, each providing a random contribution according to its own distribution at time step $1\leq t \leq T$. 
 Without knowing these distributions, a decision maker picks one arm at each time step and tries to maximize the total expected contribution in $T$ time steps. 
  When the decision maker can select multiple arms at each time, the problem is often referred to as  combinatorial multi-armed bandit (CMAB) in literature \citep{chen2016combinatorial,kveton2015tight}.  (C)MAB captures a fundamental tradeoff in most learning problems: \textit{exploration} vs. \textit{exploitation}. A common metric to evaluate the performance of (C)MAB learning algorithms is regret, which captures the difference between  the optimal expected value assuming the distributions are known and the  expected value achieved by the online learning algorithm. It is desirable to design online algorithms with sublinear $o(T)$ regrets, which roughly indicates that  the learning algorithm eventually learns the optimal solution.
  

 
  Though there have been studies on DR via  (C)MAB, most literature aims at maximizing the load reduction \citep{wang2014adaptive,lesage2017multi,jain2014multiarmed}. There is a lack of efforts on improving the reliability of CMAB algorithms for DR as well as the theoretical reliability guarantees. 
  

\subsection{Our Contributions} 
In this paper, we formulate the DR as a CMAB problem with a reliability objective, i.e. we aim to minimize the deviation between the actual total load adjustment and a target signal. The target might be caused by a sudden change of renewable energy or a peak load reduction event.   We consider a large number of residential users, and each user can commit one unit of load change (either reduction or increase) with an unknown probability. The task of the DR aggregator is to select a subset of the users to guarantee the actual load adjustment to be as close to the target as possible. The number of  users to select is not fixed, giving flexibility to the aggregator for achieving different target levels.  

In order to design our  online learning algorithm, we first develop an offline combinatorial optimization algorithm that selects the optimal subset of the users when the user behavior models are known. Based 
on the structure of the offline algorithm, we propose an online algorithm CUCB-Avg (Combinatorial Upper Confidence Bound-Average) and provide a rigorous regret analysis.  We show that, over $T$ time steps,  CUCB-Avg achieves $O(\log T)$ regret given a static target and $O(\sqrt{T\log(T)}$ regret given a time-varying target. The regrets in both cases depend polynomially on the number of users $n$. We also conduct numerical studies using synthetic DR data, showing that the performance of CUCB-Avg is much better than the classic algorithm CUCB \citep{kveton2015tight,chen2016combinatorial}, and also better than  Thompson sampling. 
In addition, we numerically show that, with minor modifications,  CUCB-Avg can cope with more realistic  behavior models with user fatigue.

Lastly, we would like to mention that though the DR model considered in this paper is very simple, the model is motivated by real pilot studies of residential DR programs, and the results have served as a guideline for designing the learning protocols \citep{thinkeco}. Besides, since real-world DR programs vary a lot among each other (depending on the DR company, local policies, reward schemes, data infrastructure, etc.), abstracting the DR model can be useful for a variety of DR programs by providing some common insights and general guidelines. When designing algorithms for real DR programs, we could modify the vanilla method to suit different specific requirements. 
 Furthermore, our algorithm design and  theoretical analysis based on the simple model may also provide insights for  other societal system applications.
 

\subsection{Related Work}
\textit{Combinatorial multi-armed bandits.}
 Most literature in CMAB studies a classic formulation which aims to maximize the total (weighted) contribution of $K$ arms with a \textit{fixed} integer $K$ (and known weights) \citep{bubeck2012regret,kveton2015tight}. 
There are also papers considering more general reward functions, for example, \citep{chen2016combinatorial} considers   objective functions that are \textit{monotonically nondecreasing}  with the parameters of the selected arms and designs Combinatorial Upper Confidence Bound (CUCB) using the principle of \textit{optimism in the face of uncertainty}.
However, the reliability objective of our CMAB problem does not satisfy the monotonicity assumption, thus the study of CUCB  cannot be directly applied here. 
Another line of work follows the \textit{Bayesian} approach and studies Thompson sampling \citep{gopalan2014thompson,wang2018thompson}. \red{However, the regret bound of Thompson sampling consists of a term that is independent of $T$ but depends  exponentially on the number of   arms $K$ in the optimal subset \citep{wang2018thompson}. Further,  \citep{wang2018thompson} shows that the exponential dependence is unavoidable. In  the residential DR problems, $K$ is usually large, so Thompson sampling may generate poor performance especially when $T$ is not very large, which is consistent with our numerical results in Figure \ref{fig: compare TS with diff. n} in Section \ref{sec: simulation}.} Finally, there is a lack of analysis on time-varying objective functions, but in many real-world applications the objectives change with time, e.g., the DR target would depend on the time-varying renewable generation. Therefore, either the learning algorithms or the theoretical analysis in literature do not directly apply to our CMAB problem, motivating the work of this paper.

\textit{Risk-aversion MAB.} There is a related line of  research on reducing risks in MAB by  selecting the \textit{single} arm with the best \textit{return-risk tradeoff} \citep{sani2012risk,vakili2016risk}.  However, there is a  lack of studies on selecting \textit{a subset of} arms so that the total contribution of the selected arms is \textit{close to a certain target.\footnote{\red{In Appendix \ref{append: risk averse}, we provide an algorithm based on the risk-aversion MAB ideas and provide numerical results.}}}
	

\textit{Learning-based demand response.} In addition to the demand response program considered in this paper and \citep{wang2014adaptive,lesage2017multi,thinkeco,psegcoolcustomer,conEdison}, where customers are directly selected by the aggregator to  perform demand response, there is a different type of DR programs based on dynamic pricing,  where the goal is to design time-varying electricity prices to automatically incentivize  desirable load reduction behaviors from the consumers  \citep{dynamicpricing}. Learning-based algorithms are also proposed for this type of DR programs to deal with, for example, the unknown utility functions of the consumers \citep{khezeli2017risk,li2017distributed,moradipari2018learning}.

\textit{Preliminary work.} Some preliminary work was presented in the conference paper \citep{li2018learning}. This journal version strengthens the regret bounds, especially for the time-varying  target case, conducts  more intensive numerical analysis using realistic data from ISOs, provides more complete proofs, and adds more  intuitions and discussions to both theoretical and numerical results.


\vspace{6pt}
\noindent\textbf{Notations.} Let $\bar E$ and $|E|$ be the complement and the cardinality of the set $E$ respectively. For any positive integer $n$, let $[n]=\{1,\dots, n \}$. Let $I_E(x)$ be the indicator function: $I_E(x)=1$ if $x\in E$ and $I_E(x)=0$ if $x\not \in E$. For any two sets $A, B$, we define $A-B\coloneqq\{x\mid x\in A, x\not \in B \}$. When $k=0$, let $\sum_{i=1}^k a_i=0$ for any $a_i$, and define the set $\{\sigma(1),\dots, \sigma(k) \}=\emptyset$ for any $\sigma(i)$. 
 For $x\in \R^k$,  we consider $\|x\|_\infty = \max_{i\in[k]}|x_i|$, and
 write $f(x)=O(g(x))$ as $\|x\|_\infty \to +\infty$  if there exists a constant $M$ such that $|f(x) | \leq M |g(x)|$ for any $x$ with $x_i\geq M$ for some $i$; and  $f(x)=o(g(x))$  if $ f(x)/g(x)\to0$ as $\|x\|_\infty \to +\infty$. We usually omit ``as $\|x\|_\infty\to +\infty$" for simplicity. For the asymptotic behavior near zero,  consider the inverse of $\|x\|_\infty$.

\section{Problem Formulation}

Motivated by the discussion above, we formulate the demand response (DR)  as a CMAB problem in this section. We  focus on load reduction to illustrate the problem. The load increase can be treated in the same way. 


Consider a DR program with an aggregator and $n$ residential customers  over $T$ time steps, where each time step corresponds to one DR event.\footnote{The specific definition of DR events and the duration of each event are up to the choice of the system designer. Our methods can accommodate different scenarios.} Each customer is viewed as an arm in our CMAB problem. We consider a simple user (customer) behavior model, where each customer may either respond to a DR event by reducing one unit of power consumption with probability $0\leq p_i\leq 1$, or not respond with probability $1-p_i$. We denote the demand reduction by customer $i$ at time step $t$ as $X_{t,i}$, which is assumed to follow Bernoulli distribution, $X_{t,i}\sim \text{Bern}(p_i)$, and is independent across time.\footnote{For simplicity, we only consider that each customer has one unit to reduce. Our learning method can be extended to multi-unit setting and/or the setting where different users have different sizes of units. But the regret analysis will be more complicated which we leave as future work. As mentioned before, results in the paper have been used as a guideline for DR field studies \citep{conEdison}.} Different customers behave independently and may respond to the same DR event with  different probabilities. Though this behavior model may be oversimplified by neglecting the influences of temperatures, humidities, user fatigue, changes in lifestyles, etc., this simple model allows us to provide useful insights on improving the  reliability of the DR programs and lay the foundation for future  research on more realistic behavior models.

At each time $1\leq t\leq T$, there is a DR event with a nonnegative demand reduction target $D_t$ determined by the power system. This reduction target might be caused by a sudden drop of renewable energy generation or a peak load reduction request, etc. 
The aggregator aims to select a subset of customers, i.e. $S_t \subseteq [n]$, such that the total demand reduction is as close to the target as possible. The cost at time $t$ can be captured by the \textit{squared deviation} of the total reduction from the target $D_t$: 
$$ L_t(S_t)= \left(\sum_{i\in S_t}X_{t,i}-D_t\right)^2.$$
Noticing that the  demand reduction $X_{t,i}$ are random, we consider a goal of selecting a subset of customers $S_t$ to minimize the   squared deviation in expectation, that is,
\begin{equation}\label{equ: problem formulation}
S_t^* \subseteq \argmin_{S_t \subseteq [n]} \E\!\left[ L_t(S_t)\right].
\end{equation}
	In this paper, we will first study the scenario where the target $D$ is time-invariant (Section~\ref{sec: online alg} and \ref{sec: regret}). Then, we will extend the results to cope with time-varying targets to
	incorporate different DR signals resulted from the fluctuations of power supply and demand (Section~\ref{sec:time_varying}). 

When the response probability profile $ p=(p_1, \dots, p_n)$ is known, the problem \eqref{equ: problem formulation} is a combinatorial optimization. In Section \ref{sec: online alg}, we will provide an offline combinatorial optimization algorithm to solve the problem \eqref{equ: problem formulation}.

In reality, the response  probabilities are usually unknown. Thus, the aggregator should learn the probabilities from the feedback of the previous demand response events, then make online decisions to minimize the difference between the total demand reduction and the target $D_t$. The learning performance is measured by $\text{Regret}(T)$, which compares the total expected cost of online decisions and the optimal total expected costs in $T$ time steps:\footnote{Strictly speaking, this is the definition of pseudo-regret, because its benchmark is the optimal expected cost: $\min_{S_t \subseteq [n]}\E L_t(S_t) $, instead of the optimal cost for each time, i.e. $\min_{S_t \subseteq [n]} L_t(S_t) $.}
\begin{equation}
\begin{aligned}
&\text{Regret}(T) \coloneqq \E\left[\sum_{t=1}^T R_t(S_t)\right], \label{equ:regret def}
\end{aligned} 
\end{equation}
where $ R_t(S_t)\coloneqq L_t(S_t) - L_t(S^*_t)$ and the expectation is taken with respect to  $X_{t,i}$ and the possibly random $S_t$.

The feedback of previous demand response events includes the responses of every selected customer, i.e., $\{X_{t,i}\}_{i\in S_t}$. Such feedback structure is called \textit{semi-bandit} in literature \citep{chen2016combinatorial}, and carries more information than bandit feedback which only includes the realized cost $L_t(S_t)$.

Lastly, we note that our problem formulation can be applied to other applications beyond demand response. One example is introduced below.

\begin{example}
	Consider a crowd-sourcing related problem. Given  budget $D_t$, a survey planner sends out surveys and offers one unit of reward for each participant. Each potential participant may  participate with probability $p_i$. Let $X_{t,i}=1$ if agent $i$ participates; and $X_{t,i}=0$ if agent $i$ ignores the survey. The survey planner intends to maximize the total number of responses without exceeding the budget too much. One possible formulation 
	 is to select subset $S_t$ such that the total number of responses is close to the budget $D_t$, 
	 $$ \min_{S_t} \E\left(\sum_{i\in S_t}X_{t,i}-D_t\right)^2.$$
	 Since the participation probabilities are unknown, the  planner can learn the participation probabilities from the previous actions of the selected agents and then try to minimize the total costs during the learning process. 
\end{example}

\section{Algorithm Design} \label{sec: online alg}
This section considers time-invariant target $D$. We will first  provide an optimization algorithm for the offline problem, then introduce the notations for online algorithms and discuss two simple algorithms: greedy algorithm and CUCB. Finally, we introduce our online algorithm CUCB-Avg. 
\subsection{Offline Optimization}
When the probability profile $p$ is known, the problem \eqref{equ: problem formulation} becomes a combinatorial optimization problem:
\begin{align}\label{equ: offline minimization}
\min_{S \subseteq [n]} \left[(\sum_{i\in S}p_i-D)^2+ \sum_{i\in S}p_i(1-p_i)\right],
\end{align}
where we omit the subscript $t$ for simplicity of notation.
Though combinatorial optimization is NP-hard  and only has approximate algorithms in general, we are able to design a simple algorithm in Algorithm \ref{alg: offline oracle} to solve the problem \eqref{equ: offline minimization} exactly. Roughly speaking, Algorithm \ref{alg: offline oracle} takes two steps: i) rank the arms according to $p_i$, ii) determine the number $k$ according to the probability profile $p$ and the target $D$ and select the top $k$ arms. The output of Algorithm~\ref{alg: offline oracle} is denoted by $\phi(p,D)$ which is a subset of $[n]$. In the following theorem, we show that such algorithm finds an optimal solution to  \eqref{equ: offline minimization}. 

\begin{algorithm}\caption{Offline optimization algorithm}\label{alg: offline oracle}
	\begin{algorithmic}[1]
		\State \textbf{Inputs:}  $p_1,\dots,  p_n\in [0,1]$, $D>0$.
		\State Rank $p_i$ in a non-increasing order:

		$ p_{\sigma(1)}\geq \dots \geq p_{\sigma(n)}$.
		\State Find the smallest $k\geq 0$ such that
$$
		\sum_{i=1}^k p_{\sigma(i)}> D-1/2.$$
Let $k=n$ if
$
		 \sum_{i=1}^{n} p_{\sigma(i)}\leq D-1/2
$. Ties are broken randomly.
		\State \textbf{Ouputs}:  $\phi(p,D)=\{\sigma(1),\dots, \sigma(k) \}$
	\end{algorithmic}
\end{algorithm}
\begin{theorem}\label{thm: offline opt optimality}
	For any $D>0$, the output of
	Algorithm \ref{alg: offline oracle}, $\phi(p, D)$, is an optimal solution to \eqref{equ: offline minimization}.
\end{theorem}
\textit{Proof Sketch.}
We defer the detailed proof to Appendix \ref{aped: offline opt} and only introduce the intuition here. An optimal set $S$  roughly has two properties: i) the total expected contribution of $S$, $\sum_{i\in S} p_i$, is closed to the target $D$, ii) the total variance of  arms in $S$ is minimized. i) is roughly guaranteed by Line 3 of Algorithm \ref{alg: offline oracle}: it is easy to show that $|\sum_{i\in\phi(p,D)} p_i-D|\leq 1/2$. ii) is roughly guaranteed by only selecting arms with higher response probabilities, as indicated by Line 2 of Algorithm \ref{alg: offline oracle}. The intuition is the following. Consider an arm with large parameter $p_1$ and two arms with smaller parameters $p_2, p_3$. For simplicity, we let $p_1=p_2+p_3$. Thus replacing $p_1 $ with $p_2,p_3$ will not affect the first term in \eqref{equ: offline minimization}. However, 
$ p_1(1-p_1) \leq p_2(1-p_2) + p_3(1-p_3)$
by $p_1^2=(p_2+p_3)^2 \geq p_2^2+p_3^2$.
Hence, replacing one arm with higher response probability by two arms with lower response probabilities will  increase the variance. 
\qed

\begin{corollary}\label{cor: D<1/2}
	When $D<1/2$, the empty set is optimal. 
\end{corollary}
\begin{remark}
	There might be more than one optimal subset. Algorithm \ref{alg: offline oracle} only outputs one of them.
\end{remark}

\subsection{Notations for Online Algorithms}

Let $\bar p_i(t)$ denote the sample average of parameter $p_i$ by time $t$ (including time $t$), then 
$$ \bar p_i(t) = \frac{1}{T_i(t)}\sum_{\tau \in I_i(t)} X_{\tau, i},$$
where  $I_i(t)$ denotes the set of time steps when arm $i$ is selected by time $t$ and $T_i(t)=| I_i(t)| $ denotes the number of times that arm $i$ has been selected by time $t$. 
Let $\bar p(t) = (\bar p_1(t), \dots, \bar p_n(t))$.
Notice that before making decisions at time $t$, only $\bar p(t-1)$ is available.

\subsection{Two Simple Online Algorithms: Greedy Algorithm and CUCB}\label{subsec: CUCB}
Next, we introduce two simple methods: greedy algorithm and CUCB, and explain their poor performance in our problem to gain intuitions for our algorithm design.

Greedy algorithm  uses the sample average of each parameter $\bar p_i(t-1)$ as an estimation of the unknown probability $p_i$ and chooses a subset based on the offline oracle described in Algorithm~\ref{alg: offline oracle}, i.e. $S_t = \phi(\bar p(t-1),D)$. The greedy algorithm is known to perform poorly  because  it
only exploits the current information, but fails to explore the unknown information, as demonstrated below. 
\begin{example}
Consider two arms that generate Bernoulli rewards with  expectation $p_1>p_2>0$. The goal is to select the arm with the higher  reward expectation, which is arm 1 in this case. Suppose after some time steps, arm 1's history sample  average $\bar p_1(t)$ is zero, while arm 2's history average $\bar p_2(t)$ is positive. In this case, the greedy algorithm will always select the suboptimal arm 2 in the future since $\bar p_2(t)>\bar p_1(t)=0$ for all future time $t$ and arm 1's history average will remain 0 due to insufficient exploration. Hence, the regret will be $O(T)$.



\end{example}

A well-known algorithm in CMAB literature  that balances the exploration and exploitation is CUCB \citep{chen2016combinatorial}. Instead of using sample average $\bar p(t-1)$ directly, CUCB considers an upper confidence bound:
\begin{equation} \label{eq:Ui}
	 U_i(t) =\min\left( \bar p_i(t-1)+ \sqrt{\frac{\alpha \log t}{2T_i(t-1)}},1\right),
\end{equation}
{\color{black}
where $\alpha\geq 0$ is the parameter to balance the tradeoff between $\bar p_i(t-1)$ (exploitation) and $T_i(t-1)$ (exploration). The output of CUCB is $S_t= \phi(U(t), D)$.  CUCB performs well in classic CMAB problems, such as maximizing the total contribution of $K$ arms for a fixed  $K$.




However, CUCB performs poorly in our problem, as shown in Section \ref{sec: simulation}. }The major problem of CUCB is the over-estimate of the arm parameter $p$. By choosing $S_t= \phi(U(t), D)$, CUCB selects less arms than needed, which not only results in a large deviation from the target, but also discourages exploration.

\begin{algorithm}\caption{CUCB-Avg}\label{alg: ucb_avg}
	\begin{algorithmic}[1]
		\State \textbf{Notations:}  $T_i(t)$ is the number of times selecting 
		arm $i$ by  time  $t$, and $\bar p_i(t)$ is the sample average of  arm $i$ by time $t$ (both including time $t$).
		\State \textbf{Inputs:} $\alpha$, $D$.
		\State \textbf{Initialization:}  For $t=1, \ldots, \ceil{\frac{n}{\ceil{2D}}}$, select $\ceil{2D}$ arms each time until each arm has been selected for at least once.
			Let $S_t$ be the set of arms selected at time $t$. Initialize $T_i(t)$ and $\bar p_i(t)$ by the observation $\{X_{t,i} \}_{i\in S_t}$.\textcolor{black}{\footnotemark{}}
		\For{$t=\ceil{\frac{n}{\ceil{2D}}}+1, \dots, T$ }
		\State 	Compute the upper confidence bound for each  $i$
		
		$U_i(t) = \min\left( \bar p_i(t-1)+ \sqrt{\frac{\alpha \log t}{2T_i(t-1)}},1\right)$.
		
		\State 	Rank $U_i(t)$ by a non-increasing order:  
		
		$U_{\sigma(t,1)}(t)\geq \dots \geq U_{\sigma(t,n)}(t)$.
		\State Find the smallest $k_t\geq 0$ such that 
		$$ \sum_{i=1}^{k_t} \bar p_{\sigma(t,i)}(t-1)>D-1/2$$
		\qquad or let $k_t=n$ if $ \sum_{i=1}^{n} \bar p_{\sigma(t,i)}(t-1)\leq D-1/2$. 
		\State  Select $S_t=\{\sigma(t,1),\dots, \sigma(t, k_t) \}$ 
		\State Update $T_i(t)$ and $\bar p_i(t)$ by observations $\{X_{t,i} \}_{i\in S_t}$
		\EndFor
	\end{algorithmic}
\end{algorithm}\footnotetext{\textcolor{black}{The initialization method  is
not unique and can be any method that selects each customer
for at least once.}}

\subsection{Our Proposed Online  Algorithm: CUCB-Avg}
Based on the discussion above, we propose a new method CUCB-Avg. The novelty of CUCB-Avg is that it utilizes both sample averages and upper confidence bounds by exploiting the structure of the offline optimal method. 

We note that the offline Algorithm~\ref{alg: offline oracle} selects the right subset of arms in two steps: i)  rank (top) arms, ii) determine the number $k$ of the top arms to select. In CUCB-Avg, we use the upper confidence bound $U_i(t)$ to rank the arms in a non-increasing order. This is the same as  CUCB. However, the difference is that our CUCB-Avg uses the sample average $\bar{p}_i(t-1)$ to decide the number of  arms to select at time $t$. The details of the algorithm are given in Algorithm~\ref{alg: ucb_avg}.

Now we explain  why the ranking rule and the selection rule of CUCB-Ave would work for our problem. 

The ranking rule is determined by $U_i(t)$. An arm with larger $U_i(t)$ is given a priority to be selected at time $t$. We note that  $U_i(t)$ is the summation of two terms: the sample average $\bar{p}_i(t-1)$ and the confidence interval radius that is related to how many times the arm has been explored. Therefore, an arm with a large $U_i(t)$ may    have a small $T_i(t-1)$, meaning  that the arm has not been explored enough; and/or  have a large $\bar p_i(t-1)$, indicating that the arm frequently responds in the history. In this way, CUCB-Avg selects both the under-explored arms (\textit{exploration}) and the arms with good performance in the past (\textit{exploitation}).

When determining $k$, CUCB-Avg uses the sample averages and selects enough arms such that the total sample average is close to $D$. Compared with CUCB which uses upper confidence bounds to determine $k$, our algorithm selects more arms, which reduces the load reduction difference from the target and also encourages  exploration.


\vspace{-8pt}

\section{Regret analysis}\label{sec: regret}
In this section, we will prove that our algorithm CUCB-Avg achieves $O(\log T)$ regret when $D$ is time invariant.

\subsection{The Main Result}

\begin{theorem}\label{thm: regret bdd D}
	There exists a constant $\epsilon_0>0$ determined by $p$ and $  D$, such that 
	for any $\alpha>2$, the regret of CUCB-Avg is upper bounded by
	\begin{equation}\label{equ: regret bdd D}
	\textup{Regret}(T)\leq M\left(\ceil{\frac{n}{\ceil{2D}}}+\frac{2 n}{\alpha -2}\right)+ \frac{\alpha Mn\log T}{2\epsilon_0^2}, 
	\end{equation}
	where $M=\max(D^2, (n-D)^2)$. 
\qed	
\end{theorem}

We make a few comments before the proof.

{\color{black}
\nbf{Dependence on  $T$ and $n$.} The  dependence on the horizon $T$ is $O(\log T)$, so the average regret diminishes to zero as $T$ increases, indicating that our algorithm learns the customers' response probabilities effectively. The dependence on $n$ is polynomial, i.e. $O(n^3)$ by $M\sim O(n^2)$, showing that our algorithm  can handle a large number of arms effectively. The cubic dependence is likely to be a proof artifact and improving the dependence on $n$ is left as future work.


}




{\color{black}
\nbf{Role of $\epsilon_0$.} The bound depends on a constant term $\epsilon_0$ determined by $p$ and $D$ and such a bound is  referred to as a \textit{distribution-dependent  bound} in literature. We defer the explicit expression of $\epsilon_0$ to Appendix \ref{aped: D no regret} and only explain the intuition behind $\epsilon_0$ here. Roughly, $\epsilon_0$ is a robustness measure of our offline optimal algorithm, in the sense that  if the probability profile $p$ is perturbed  by $\epsilon_0$, i.e., $ |\tilde p_i-p_i|< \epsilon_0$ for all $i$, the output $\phi(\tilde p,D)$ of Algorithm~\ref{alg: offline oracle} would still be  optimal for the true profile $p$. Intuitively, if $\epsilon_0$ is large, the learning task is easy because we are able to find an optimal subset given a poor estimation, leading to a small regret. This explains why the upper bound in (\ref{equ: regret bdd D}) decreases when $\epsilon_0$ increases.
}


To discuss what factors will affect the robustness measure $\epsilon_0$, we provide an explicit expression of $\epsilon_0$ under two assumptions in the following proposition.

\begin{proposition}\label{prop: A1A2 epsilon0}
	Consider the following  assumptions.\\
	\textup{(A1)} 	$p_i$ are positive and distinct $p_{\sigma(1)}> \dots > p_{\sigma(n)}>0$. \\
	\textup{(A2)}	There exists $k\geq 1$ such that 
	$ \sum_{i=1}^k p_{\sigma(i)}> D-1/2,$ and $ \sum_{i=1}^{k-1} p_{\sigma(i)}< D-1/2
$.\\
	Then the $\epsilon_0$ in Theorem \ref{thm: regret bdd D} can be determined by:
	\begin{equation}
	\epsilon_0 = 
	\min\left(\frac{\delta_1}{k}, \frac{\delta_2}{k}, \frac{\Delta_k}{2}\right),
	\label{equ: def epsilon}
	\end{equation}
	where $k=| \phi(p, D)|$, $\sum_{i=1}^k p_{\sigma(i)}=D-1/2+\delta_1,$ $ \sum_{i=1}^{k-1} p_{\sigma(i)}=D-1/2-\delta_2,$ and $\Delta_i =p_{\sigma(i)}-p_{\sigma(i+1)} ,  \forall \ i=1,\dots, n-1$.
\end{proposition}
We defer the proof of the proposition to Appendix \ref{aped: D no regret Ass} and only make two comments here. Firstly, it is easy to verify that Assumptions (A1) and (A2) imply $\epsilon_0>0$. Secondly, we explain why $\epsilon_0$ defined in \eqref{equ: def epsilon} is a robustness measure, that is, we show if $\forall\, i,\ | \tilde p_i- p_i| <\epsilon_0$,  then $\phi(\tilde p, D)=\phi(p,D)$. This can be proved in two steps. Step 1: when $\epsilon_0 \leq \frac{\Delta_k}{2}$,
the $k$ arms with higher $\tilde {p}_i$ are the same $k$ arms with higher $p_i$ because for any $1\leq i \leq k$ and $k+1 \leq j \leq n$, we have $
\tilde p_{\sigma(i)}> p_{\sigma(k)}-\epsilon_0 \geq p_{\sigma(k+1)} + \epsilon_0> \tilde p_{\sigma(j)}$. 
 Step 2: by $\epsilon_0 \leq \min\left( \frac{\delta_1}{k},\frac{\delta_2}{k}\right)$ and the definition of $\delta_1$ and $\delta_2$, when $|\tilde p_i-p_i|\leq \epsilon_0$ for all $i$,  we have $\sum_{i=1}^k \tilde p_{\sigma(i)}> D-1/2$ and $\sum_{i=1}^{k-1} \tilde p_{\sigma(i)}  < D-1/2$. Consequently, by Algorithm \ref{alg: offline oracle}, $\phi(\tilde p, D)=\{\sigma(1),\dots, \sigma(k) \}=\phi(p,D)$.

Finally, we briefly discuss how to generalize the expression \eqref{equ: def epsilon} of $\epsilon_0$ to cases without (A1) and (A2). When (A1) does not hold, we only consider the gap between the arms that are not in a tie, i.e. $\{\Delta_i| \ \Delta_i>0, \  1\leq i\leq n-1 \}$. When (A2) does not hold and $\sum_{i=1}^{k-1} p_{\sigma(i)}=D-1/2$, we consider  less than $k-1$ arms to make the total expected contribution below $D-1/2$. An explicit expression of $\epsilon_0$ is provided in Appendix \ref{aped: D no regret}.

{\color{black}
\nbf{Comparison with the regret bound of classic CMAB.}
In classic CMAB literature when the goal is to select $K$ arms with the highest parameters given a fixed integer $K$, the regret bound usually depends on $\frac{\Delta_K}{2}$ \citep{kveton2015tight}. We note that $\frac{\Delta_K}{2}$ is similar to  $\epsilon_0$ in our problem, as it is the robustness measure of the top-$K$-arm problem in the sense that given any estimation $\tilde p$ with estimation error at most $\frac{\Delta_K}{2}$: $\forall\, i,\ | \tilde p_i- p_i| <\frac{\Delta_K}{2}$,  the top $K$ arms with the profile $\tilde p$ are the same as that with the profile $p$. In addition, we would like to mention that the regret bound in literature is usually linear on $1/\Delta_K$, while our regret bound is $1/\epsilon_0^2$. This difference may be an artificial effect of the proof techniques because our CMAB problem is more complicated. We leave it as future work to strengthen the results. 
}



\subsection{Proof of Theorem \ref{thm: regret bdd D}}

\nbf{Proof outline:} We divide the $T$ time steps into four parts, and bound the regret in each part separately. The partition of the time steps are  based on event $E_t$ and the event $B_t(\epsilon_0)$ defined below.
Let $E_t$ be the event when the sample average is outside the confidence interval considered in Algorithm \ref{alg: ucb_avg}: 
$$ E_t \coloneqq \left\{ \exists \, i\in [n], \ |\bar p_i(t-1) -p_i|\geq \sqrt{\frac{\alpha \log t}{2 T_i(t-1)}}   \right\}.$$
For any $\epsilon>0$, let $B_t(\epsilon)$ denote the event when Algorithm \ref{alg: ucb_avg} selects an arm who has been explored for no more than $ \frac{\alpha \log T}{2\epsilon^2} $ times:
\begin{equation}
B_t(\epsilon)\coloneqq \left\{ \exists \, i\in S_t, \ s.t. \ T_i(t-1) \leq \frac{\alpha \log T}{2\epsilon^2}   \right\}. \label{eq:Btepsilon}
\end{equation} 
Let $\epsilon_0>0$ be a  small number  such that Lemma \ref{lem: no regret bar E bar B} holds.

Now, we will define the four parts of the $T$ time steps, and briefly introduce the regret bound of each part.

\begin{enumerate}
	\item \textit{Initialization:} the regret bound does not depend on $T$ (Inequality (\ref{eq: D Rt(St)<= M})).
	\item \textit{When event $E_t$ happens:} the regret  bound does not depend on $T$ because $E_t$ happens rarely due to concentration properties in statistics (Lemma \ref{lem: regret 1bdd}).
	\item \textit{When event $\bar E_t$ and $B_t(\epsilon_0) $ happen:} the regret is at most $O(\log T)$ because $B_t(\epsilon_0)$ happens for at most $O(\log T)$ times (Lemma \ref{lem: regret4 bdd}).
	\item \textit{When  event $\bar E_t$ and $\bar B_t(\epsilon_0) $ happen,} the regret is zero due to the  enough exploration of the selected arms (Lemma \ref{lem: no regret bar E bar B}).
\end{enumerate}
Notice that the time steps are not divided sequentially here. For example, it is possible that $t=10$ and $t=30$ belong to Part 2 while $t=10$ belongs to Part 3.

\nbf{Proof details:} Firstly, it is without loss of generality to require $D\geq 1/2$ because when $D<1/2$, the optimal set is known to be the empty set by Corollary \ref{cor: D<1/2}, so the regret is zero by selecting no customers.

Secondly, we note that for all time steps $1\leq t\leq T$ and any $S_t\subseteq[n]$, the regret at   $t$ is upper bounded by 
\begin{align}
R_t(S_t )&  \leq L_t(S_t)      \leq  \max(D^2, (n-D)^2)  \eqqcolon M. \label{eq: D Rt(St)<= M}
\end{align}
Thus, the regret of initialization (Part 1) at $t=1,\dots, \ceil{\frac{n}{\ceil{2D}}}$ is  bounded by $M\ceil{\frac{n}{\ceil{2D}}}$.

Next, we bound the regret of Part 2 by the  Chernoff-Hoeffding's concentration  inequality. The intuition behind the proof is that $E_t$ happens rarely because the sample average $\bar{p}_i(t)$ concentrates around the true value $p_i$ with a high probability. 
\begin{theorem}[Chernoff-Hoeffding's inequality]\label{thm: chernoff}
Consider i.i.d. random variables $X_1,\dots, X_m$ with support  $[0,1]$ and mean $\mu$, then we have
	\begin{equation}
	\Pb\left(|\sum_{i=1}^m X_i -m\mu |\geq m \epsilon\right)\leq 2 e^{-2m \epsilon^2} .
	\end{equation}
\end{theorem}

\begin{lemma}\label{lem: regret 1bdd}
	
	When $\alpha>2$, we have
$$
	\E \left[\sum_{t=1}^T I_{E_t} R_t(S_t)\right]\leq \frac{2Mn}{\alpha -2}.
$$
\end{lemma}
\begin{proof} The number of times $E_t$ happens is bounded by
	\begin{align*}
	&	\E \left[\sum_{t=1}^T I_{E_t} \right]= \sum_{t=1}^T \Pb(E_t)   \\
	& \leq \sum_{t=1}^T \sum_{i=1}^n  \Pb\left(|\bar p_i(t-1)-p_i| \geq \sqrt{\frac{\alpha \log t}{2 T_i(t-1)}}\right)   \\
	& \leq \sum_{t=1}^T \sum_{i=1}^n  \sum_{s=1}^{t-1}\Pb(|\bar p_i(t-1)-p_i| \geq \sqrt{\frac{\alpha \log t}{2 s}},  T_i(t-1)  = s)   \\
	& \leq \sum_{t=1}^T \sum_{i=1}^n  \sum_{s=1}^{t-1}\, \frac{2}{t^\alpha}  \leq \sum_{t=1}^T \, \frac{2n}{t^{\alpha-1}}   \leq \frac{2n}{\alpha -2},
	\end{align*}
	where the first inequality is by enumerating possible $i \in [n]$, the second inequality is by enumerating possible values of $T_i(t-1)$: $\{1, \dots, t-1\}$, the third inequality is by Chernoff-Hoeffding's inequality, and the last inequality is by $\sum_{t=1}^T \frac{1}{t^{\alpha-1}}  \leq \int_{1}^{+\infty} \frac{1}{t^{\alpha-1}} \leq \frac{1}{\alpha -2}$. Then by inequality \eqref{eq: D Rt(St)<= M} the proof is completed.
\end{proof}

Next, we show the regret of Part 3 is at most  $O(\log T)$.
\begin{lemma}\label{lem: regret4 bdd}
	For any $\epsilon_0>0$, the regret in Part 3 is bounded by
	$
	\E\left[ \sum_{t=1}^T  R_t(S_t)I_{\{ \bar E_t, B_t(\epsilon_0) \}}\right] \leq \frac{\alpha M n \log T}{2\epsilon_0^2}
	$
\end{lemma}
\begin{proof}
	By the definition of $B_t(\epsilon_0)$ in \eqref{eq:Btepsilon}, whenever $B_t(\epsilon_0)$ happens, the algorithm  selects an arm $i$ that has not been selected for $\frac{\alpha \log T}{2\epsilon_0^2}$ times, increasing the selection time counter $T_i(t)$ by one. Hence, $ B_t(\epsilon_0) $ can happen for at most $\frac{\alpha n \log T}{2\epsilon_0^2}$ times. Then, by inequality \eqref{eq: D Rt(St)<= M}, the proof is completed.
\end{proof}

When $\bar E_t$ and $\bar B_t(\epsilon_0)$ happen (Part 4),  every selected arm is fully explored and every arm's sample average is within the confidence interval. As a result,  CUCB-Avg selects the right subset and hence contributes zero regret. This is formally stated in the following lemma.

\begin{lemma}\label{lem: no regret bar E bar B}
	
	There exists $\epsilon_0>0$, such that for each $1\leq t\leq T$, if $\bar E_t$ and $\bar B_t(\epsilon_0)$ happen,  CUCB-Avg selects an optimal subset and $\E\!\left[ R_t(S_t)I_{\{ \bar E_t,\bar B_t(\epsilon_0) \}} \right]=0$. Consequently, the regret in Part 4 is 0.
\end{lemma}
\nit{Proof Sketch:} We defer the proof to Appendix \ref{aped: D no regret Ass} and \ref{aped: D no regret} and sketch the proof ideas here, which is based on two facts: \\
Fact 1: when $\bar E_t$ and $\bar B(\epsilon_0)$ happen, the upper confidence bounds can be bounded by
$ U_i(t)>p_i$ for all $i\in [n]$,
and the confidence bounds of the selected arm $j$ satisfy
$$ 
\left|\bar p_j(t-1) -p_j\right|<\epsilon_0, \ U_j(t)<p_j + 2\epsilon_0, \ \forall\, j \in S_t.$$
Fact 2: when $\epsilon_0$ is small enough, CUCB-Avg  selects an optimal subset.

To get the intuition for Fact 2, we consider the expression of $\epsilon_0$ in \eqref{equ: def epsilon} under Assumption (A1) (A2) in Proposition \ref{prop: A1A2 epsilon0}. Let  $\phi(p,D)=\{\sigma(1),\dots, \sigma(k)  \}$ denote the  optimal subset. In the following, we roughly explain why the selected subset $S_t$ is optimal given $\epsilon_0$ defined in \eqref{equ: def epsilon}:

	i) By $\epsilon_0 \leq \frac{\Delta_k}{2}$, we can show that the selected subset $S_t$ is either a superset or a subset of the optimal subset $\{\sigma(1),\dots, \sigma(k)  \}$.
	
	ii) By $\epsilon_0 \leq \delta_1/k$, we can show that we will not select more than $k$ arms, because, informally, even if  we underestimate $p_i$, the sum of arms in $\{\sigma(1),\dots, \sigma(k)  \}$ is still larger than $D-1/2$.
	
iii) By  $\epsilon_0 \leq \delta_2/k$, we can show that we will not select less than $k$ arms, because, informally, even if  we overestimate $p_i$, the sum of  $k-1$ arms in $\{\sigma(1),\dots, \sigma(k)  \}$ is still smaller than $D-1/2$.
\qed


The proof of Theorem \ref{thm: regret bdd D} is completed by summing up the regret bounds of Part 1-4.

\section{Time-varying target}\label{sec:time_varying}
In practice,  the load reduction target is usually time-varying. We will study the performance of CUCB-Avg in the time-varying case below. 


{\color{black}
Notice that CUCB-Avg can be directly applied to the time-varying case by using $D_t$ in Algorithm \ref{alg: ucb_avg} at each time step $t$. }

Next, we  provide a regret bound for CUCB-Avg in the time-varying case. Notice that we impose no assumption on $D_t$ except that it is bounded, which is almost always the case in practice.
\begin{assumption}\label{ass: Dt bdd}
	There exists a finite $\bar{D}>0$ such that
	$0< D_t \leq \bar D, \ \forall \ 1\leq t\leq T$.
\end{assumption}

\begin{theorem}\label{thm: Dt regret bound}
	Suppose Assumption \ref{ass: Dt bdd} holds. When $T>2$, for  any $\alpha>2$, the regret of CUCB-Avg is  bounded by
	\begin{align*}
	\textup{Regret}(T)&\leq \bar M  n+ \frac{2\bar Mn}{\alpha -2} 
	+ \frac{\alpha \bar M n \log T}{2\epsilon_1^2}\\
	& + 2n^2 \sqrt{2\alpha \log T}\sqrt{T+ \frac{\alpha  \log T}{2\epsilon_1^2}},
	\end{align*}
	where  $\bar M = \max(\bar D^2, n^2)$, $\epsilon_1= \min( \frac{\Delta_{\text{min}}}{2}, \frac{\beta_{\text{min}}}{n})$, $\Delta_{\text{min}}= \min\{\Delta_i \mid 1\leq i \leq n-1, \Delta_i>0\}$ and $\beta_{\text{min}} = \min\{p_i \mid 1\leq i \leq n , p_i>0\}$.
\end{theorem}


Before the proof, we make a few comments below.

{\color{black}
\nbf{Dependence on $T$.} The bound is $O(\sqrt{T \log T})$, still sublinear in $T$, meaning that our algorithm learns the customers' response probabilities well enough to yield diminishing average regret in the time-varying case.


The  dependence on $T$ is worse than the static case which is $O(\log T)$. We briefly discuss the intuition behind this  difference.
In the proof of Theorem \ref{thm: regret bdd D}, we show that there exists a threshold $\epsilon_0$ depending on $D$ such that when the estimation errors of parameter $p_i$ for $i\in S_t$ are below $\epsilon_0$, our algorithm selects the optimal subset (Lemma \ref{lem: no regret bar E bar B}). Moreover, we also show that as $t$ increases, with high probability the estimation error will decrease and eventually our algorithm will find the optimal subset and generate no more regret. 
However, in the time-varying case the argument above no longer holds because the threshold $\epsilon_0$ will change with $D_t$, denoted as $\epsilon_0(D_t)$, and it is possible that the estimation error will always be larger than $\epsilon_0(D_t)$, as a result the algorithm may not find the optimal subset with high probability even when  $t$ is large. This roughly explains why the bound of the time-varying case  is worse than that of the static case. 


In addition, we provide some intuitive explanation for the scaling  $O(\sqrt{T \log T})$. It can be shown that  the regret at time $t$ is almost bounded by the estimation error at time $t$ under some conditions (Lemma \ref{lem: Sst under assumption Dt}). Since the estimation error is roughly captured by our confidence interval in \eqref{eq:Ui}, which  scales like $O(\sqrt{\log T/t})$,  the total regret scales like $ \sum_{t=1}^T O(\sqrt{\log T/t}) = O(\sqrt{T\log T})  $.

}
Finally, we note that the regret bound is for the worst-case scenario and  the  regret in practice may be smaller.

\nbf{Dependence on $n$.} The bound is polynomial on the number of arms $n$: $O(n^3)$ by $\bar M \sim O(n^2)$, demonstrating that our algorithm  can learn a large number of arms effectively in the time-varying case. Improving the cubic dependence on $n$ is left as future work.

\nbf{Role of $\epsilon_1$.} Notice that $\epsilon_1$ only depends on $p$ and does not depend on the target $D_t$. Roughly speaking, $\epsilon_1$ captures how difficult it is to rank the arms correctly by the value of $p_i$, in the sense that as long as the estimation error of each $p_i$ is smaller than $\epsilon_1$, the rank based on the estimation will be the correct rank based on the true parameter $p_i$. 




\textcolor{black}{\subsection{Proof of Theorem \ref{thm: Dt regret bound}}}
Most parts of the proof is similar to the static case. We also consider $D_t\geq 1/2$ without loss of generality due to Corollary \ref{cor: D<1/2}. Besides, we also  divide the time steps into four parts and complete the proof by summing up the regret bound of each part. The first three parts can be bounded in  similar ways as the static case. The major difference comes from the Part 4.

(1) Initialization: the regret can be bounded by $\bar M n$ because the initialization at most lasts for $n$ time steps and $\bar M$ is an upper bound of the single-step regret.

(2) When $E_t$ happens: 
notice that Lemma \ref{lem: regret 1bdd} still holds in the time-varying case if we replace $M$ with $\bar M$, so the second part is bounded by
$
\E \sum_{t=1}^T I_{E_t} R_t(S_t)\leq \frac{2\bar Mn}{\alpha -2}
$.

(3) When $\bar E_t$ and $B_t(\epsilon_1)$ happen:
notice that Lemma \ref{lem: regret4 bdd} still holds so
$
\E \sum_{t=1}^T  R_t(S_t)I_{\{ \bar E_t, B_t(\epsilon_1) \}} \leq \frac{\alpha \bar M n \log T}{2\epsilon_1^2}
$.

(4) When $\bar E_t$ and $\bar B_t(\epsilon_1)$ happen, we can show that the regret is $O(\sqrt{T \log T})$ as stated in the lemma below.
\begin{lemma}\label{lem: regret4 Dt bound}
	The regret in Part (4) can be bounded by
$$	\E \sum_{t=1}^T  R_t(S_t)I_{\{ \bar E_t,\bar B_t(\epsilon_1) \}} \leq 2n^2\sqrt{\frac{\alpha \log T}{2}}  \sqrt{T+\frac{\alpha \log T}{2\epsilon_1^2}}. $$
\end{lemma}

\begin{proof}


Our proof relies on the following lemma which shows that the regret at time $t$ can roughly be  bounded by the estimation error $\epsilon$ at $t$ when $\epsilon\leq \epsilon_1$.  

\begin{lemma}\label{lem: Sst under assumption Dt} 
	For any time step $t$,  consider any $D_t$ and any $0< \epsilon\leq \epsilon_1$ such that
	$\Pb(\bar E_t, \bar B_t(\epsilon))>0$.
	Let $\F_t$ denote the natural filtration up to time $t$. For any  $\F_{t-1}$ such that   $\bar E_t$ and $ \bar B_t(\epsilon)$ are true, 
 we have
$
	\E [R_t(S_t) \mid \F_{t-1}]\leq 2n \epsilon.
$ 
\end{lemma}
\nit{Proof sketch.}
Due to the space limit, we defer the proof to Appendix \ref{append: regret error if epsilon1>xi with assumption} and only sketch the proof here.  Firstly, we are able to show that under $\bar E_t$ and $\bar B_t(\epsilon)$, the selected subset differs from the optimal subset for at most one arm. This is mainly due to $\epsilon\leq \epsilon_1$. Secondly, we can bound the regret of the suboptimal selections by $O(\epsilon)$, which is mainly due to our quadratic loss function. \qed


Next, we will finish the proof by bounding the estimation errors. We introduce  event $H_t^q$ to represent that each selected arm $i$ at time $t$ has been selected for more than $
\frac{\alpha \log T}{2\epsilon_1^2}+q $ times for  $q=0, 1, 2, \dots$:
	$$H_t^q\coloneqq \left\{ \forall\, i\in S_t,  T_i(t-1) > \frac{\alpha \log T}{2\epsilon_1^2}+q \right\} \cap \bar E_t \cap \bar B_t(\epsilon_1).$$


In addition, we define the estimation error $\eta_q$ by the confidence interval radius when an arm has been explored for  $\frac{\alpha \log T}{2\epsilon_1^2}+q-1$ times:
$
\frac{\alpha \log T}{2\eta_q^2}=\frac{\alpha \log T}{2\epsilon_1^2}+q-1
$,
that is, 
$
\eta_q = \sqrt{ \frac{\frac{\alpha \log T}{2}}{q -1+ \frac{\alpha \log T}{2\epsilon_1^2}}}
$.
The proof is completed by:
	\begin{align*}
&\E\!\left[ \sum_{t=1}^T  R_t(S_t)I_{\bar E_t\cap \bar B_t(\epsilon_1)} \right] 
 = \sum_{q=1}^T \sum_{t=1}^T  \E\!\left[ R_t(S_t)I_{(H_t^{q-1} -H_t^q)}\right] \\
& \leq  \sum_{q=1}^T\sum_{t=1}^T 2n \eta_{q} \Pb(H_t^{q-1}-H_t^q) \\
& \leq  \sum_{q=1}^T 2n^2 \eta_{q} =  2n^2\sum_{q=1}^T\sqrt{ \frac{\frac{\alpha \log T}{2}}{q-1 + \frac{\alpha \log T}{2\epsilon_1^2}}}\\
& \leq 2n^2\sqrt{\frac{\alpha \log T}{2} }\int_{0}^{T}\sqrt{ \frac{1}{q-1 + \frac{\alpha \log T}{2\epsilon_1^2}}}\ \mathrm{d}q\\
& \leq 4n^2\sqrt{\frac{\alpha \log T}{2}}  \sqrt{T+\frac{\alpha \log T}{2\epsilon_1^2}},
\end{align*}
where the first  equality is by $ \bar E_t \cap \bar B_t(\epsilon_1)=\cup_{q=1}^T (H_t^{q-1}-H_t^q)$;  the first inequality is by taking conditional expectation on $H_t^{q-1}-H_t^q$ and by  Lemma  \ref{lem: Sst under assumption Dt}, $(H_t^{q-1}-H_t^q) \subseteq \bar E_t \cap \bar B_t(\eta_q)$ and $\eta_q\leq \epsilon_1$; the second one is because
$H_t^{q-1}-H_t^q\subseteq\bigcup_{i=1}^n \{i \in S_t, \, T_i(t-1) = \frac{\alpha \log T}{2\epsilon_1^2}+q \}
$
and $\{i \in S_t, \, T_i(t-1) = \frac{\alpha \log T}{2\epsilon_1^2}+q \}$ occurs at most once; the third inequality uses the fact that $T>2$ and thus $ \frac{\alpha \log T}{2\epsilon_1^2}>1$.
\end{proof}


\begin{figure}
	\begin{subfigure}[b]{0.22\textwidth}
		\includegraphics[width=\textwidth]{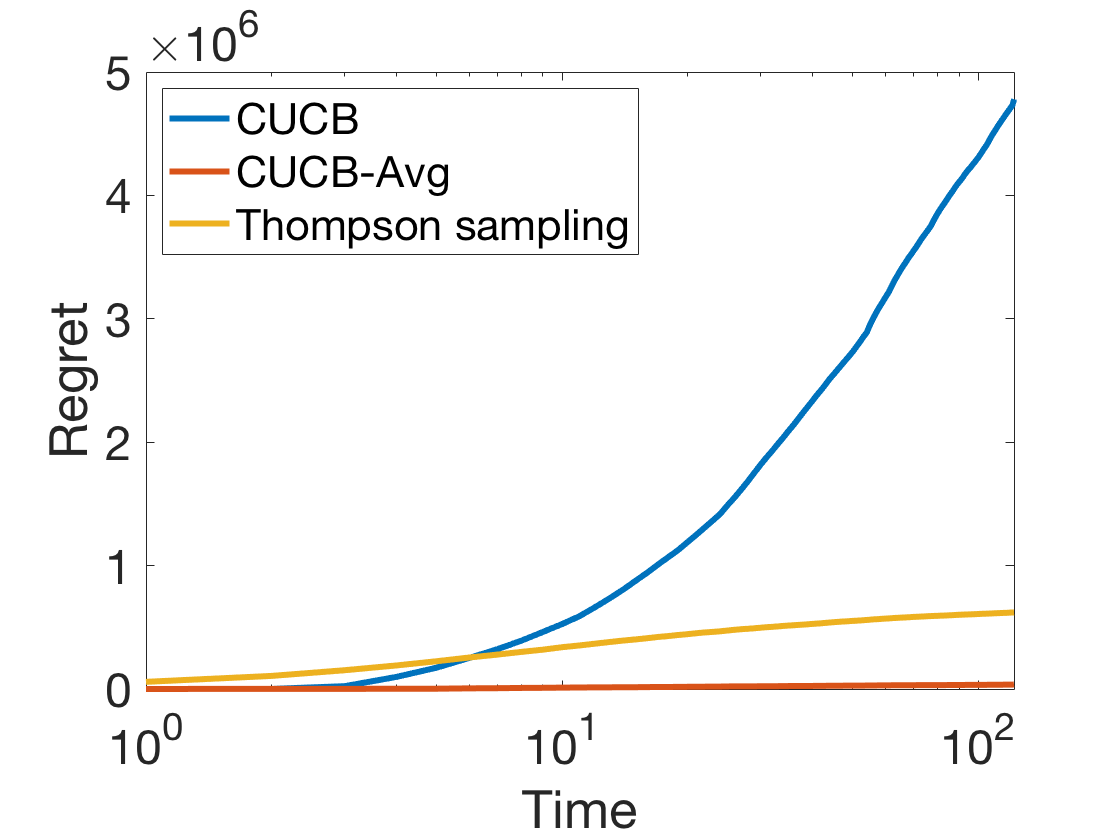}
		\caption{Average peak.}
	\end{subfigure} 
	\hfill
	\begin{subfigure}[b]{0.22\textwidth}
		\includegraphics[width=\textwidth]{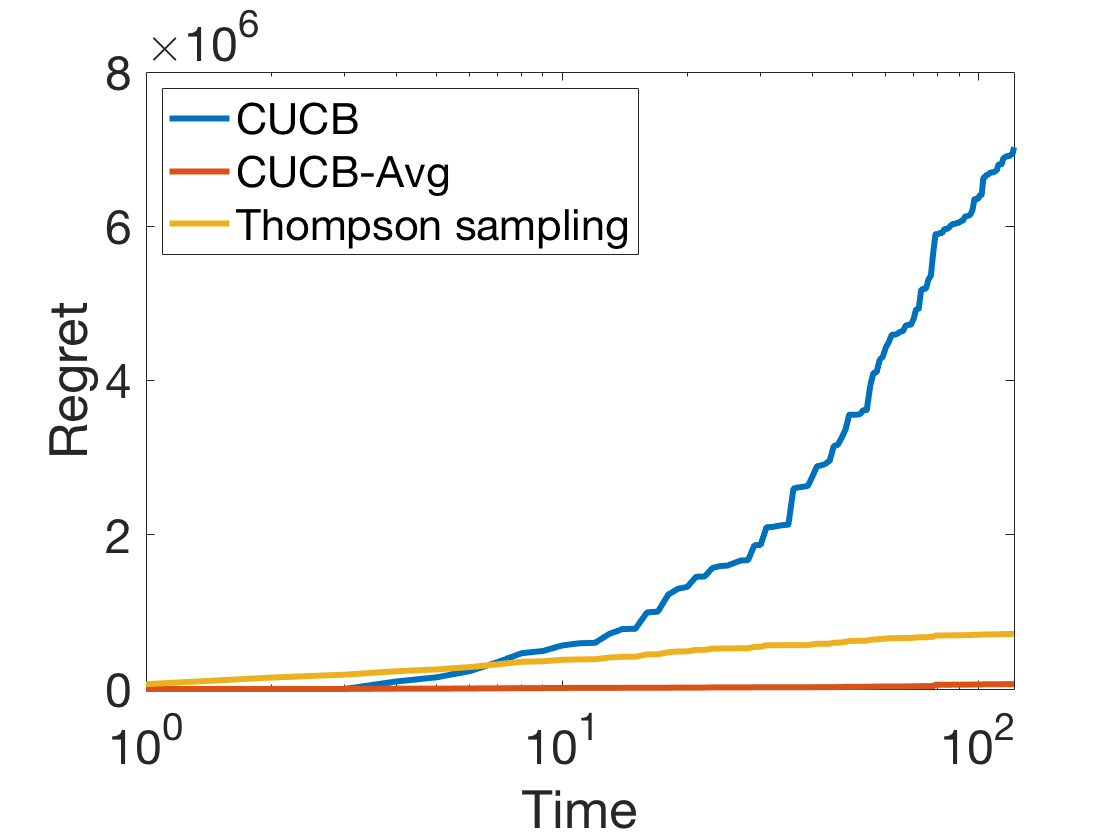}
		\caption{Daily peak.}
	\end{subfigure}
	\caption{The regret of CUCB, CUCB-Avg, and TS.}
	\label{fig: regret}
\end{figure}

\begin{figure}
	\begin{subfigure}[b]{0.22\textwidth}
		\includegraphics[width=\textwidth]{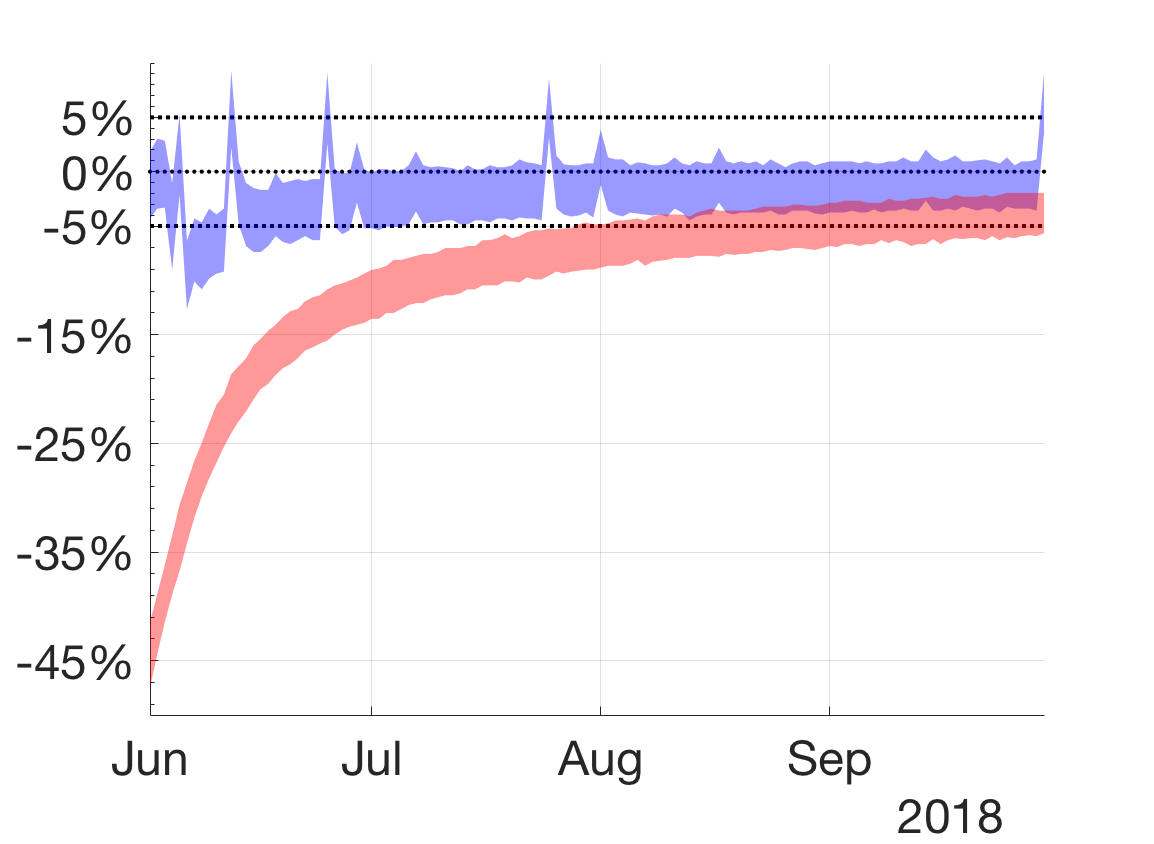}
		\caption{Average peak.}
	\end{subfigure} 
	\hfill
	\begin{subfigure}[b]{0.22\textwidth}
		\includegraphics[width=\textwidth]{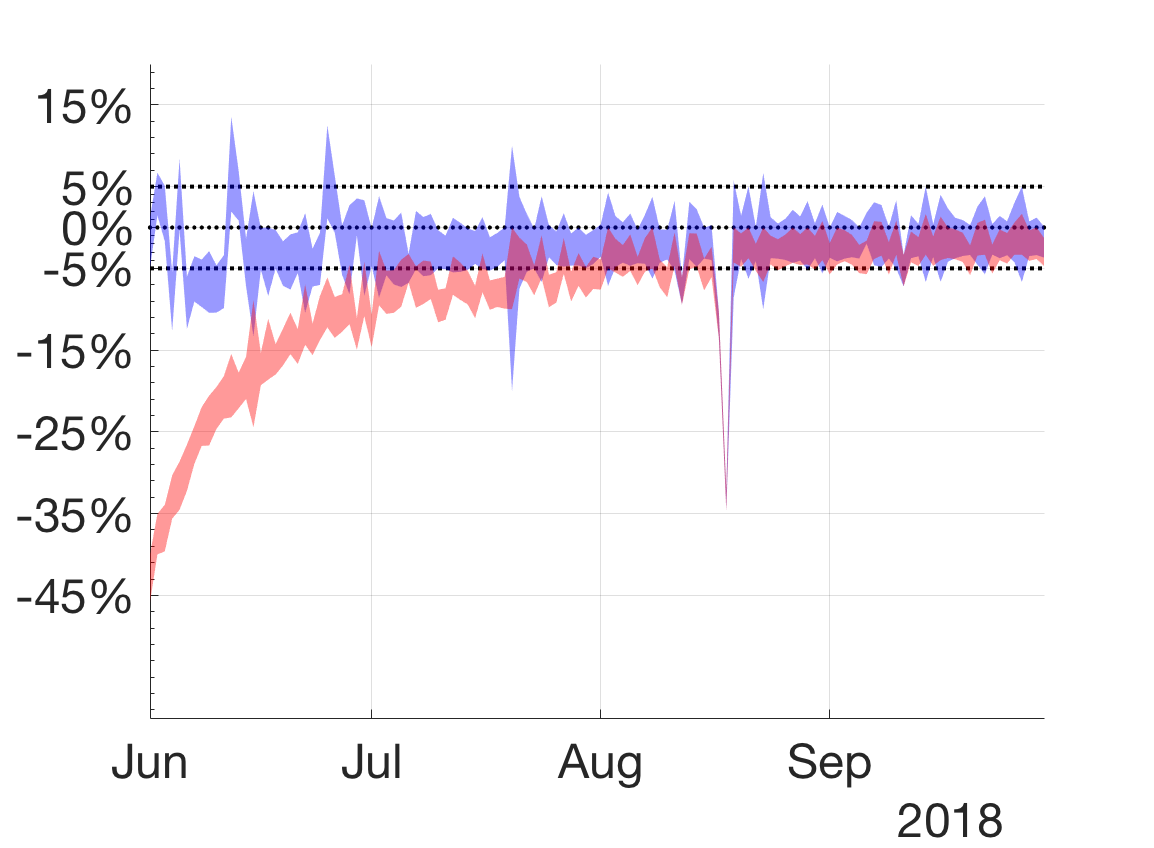}
		\caption{Daily peak.}
	\end{subfigure}
	\caption{90\% confidence intervals of  load reduction's relative errors of CUCB-Avg (blue) and Thompson sampling (red).}
	\label{fig: conf interval}
\end{figure}
\begin{figure}
	\begin{subfigure}[b]{0.22\textwidth}
		\includegraphics[width=\textwidth]{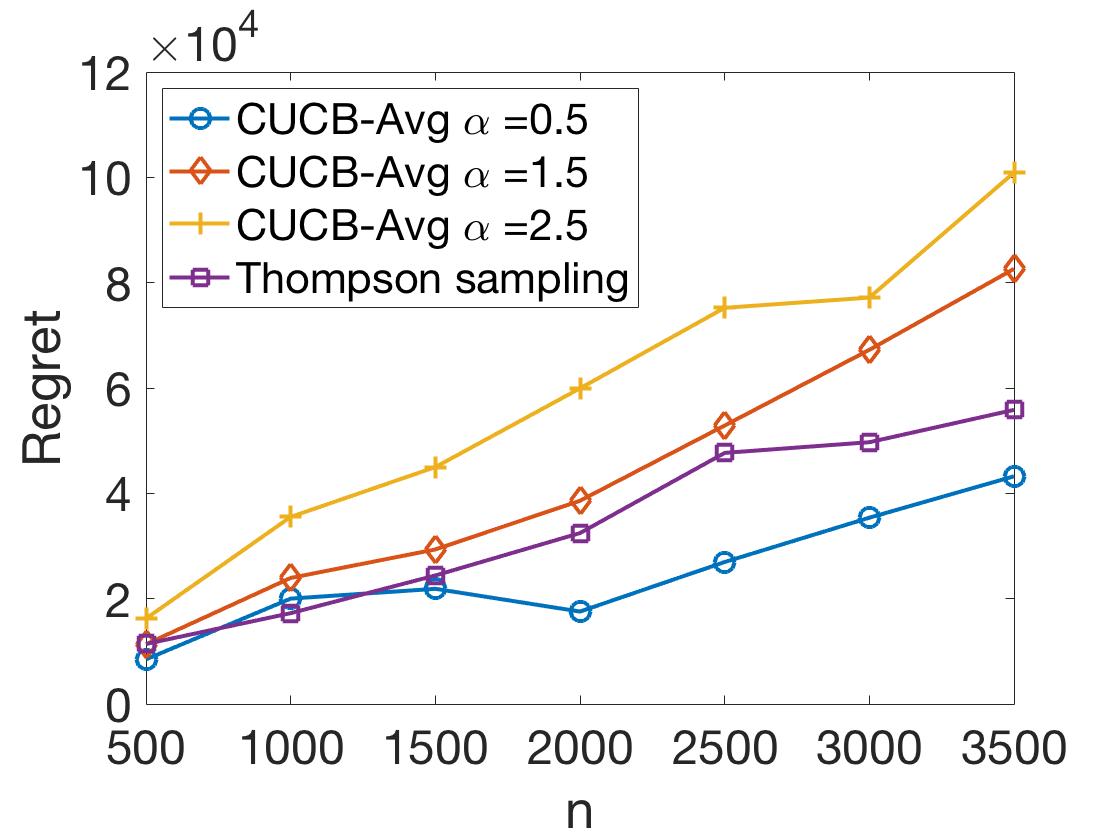}
		\caption{When $T=10^5$}
	\end{subfigure} 
	\hfill
	\begin{subfigure}[b]{0.22\textwidth}
		\includegraphics[width=\textwidth]{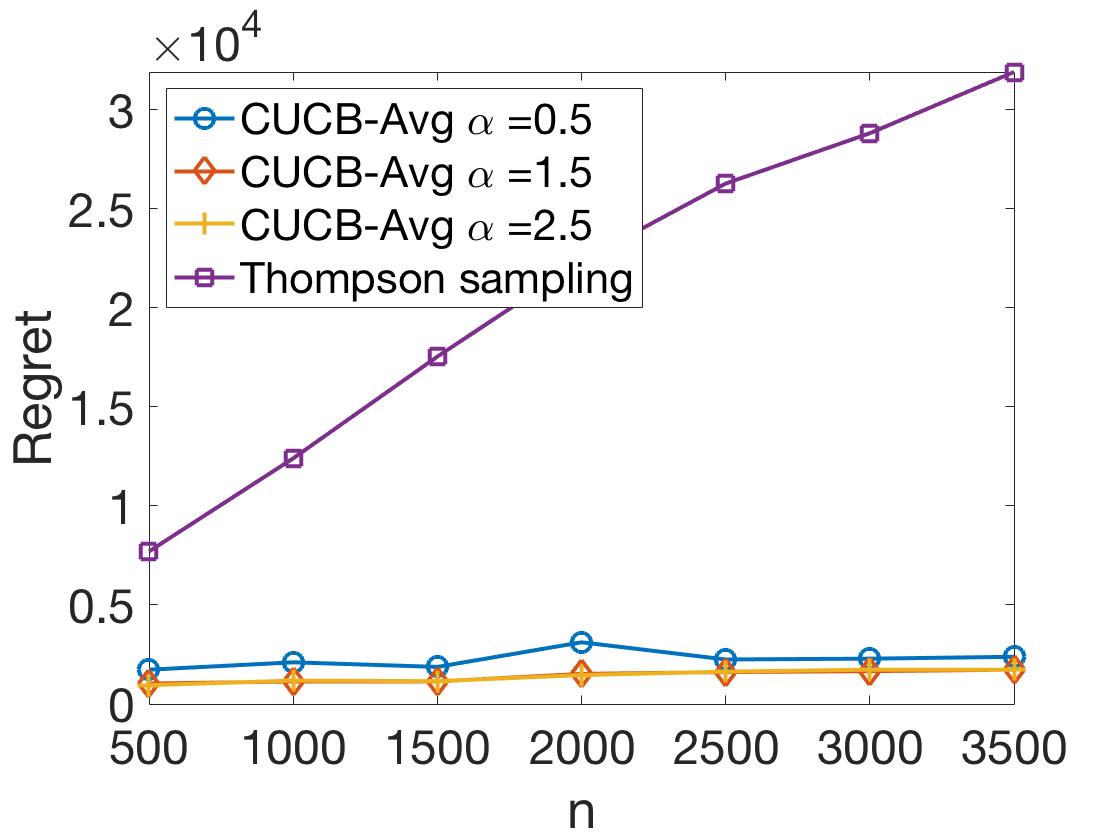}
		\caption{When $T=122$}
	\end{subfigure}
	\caption{The regret of TS and CUCB-Avg for different $n$.}
	\label{fig: compare TS with diff. n}
\end{figure}

{\color{black}\section{Numerical Experiments}\label{sec: simulation}}
In this section, we conduct numerical experiments to complement the theoretical analysis above.

\subsection{Algorithms comparison}\label{subsec: numerical Dt}
We will compare our algorithm with CUCB  \citep{chen2016combinatorial}, which is briefly explained in Section \ref{subsec: CUCB}, and Thompson sampling (TS), an algorithm with good empirical performance in classic MAB problems. In TS, the unknown parameter $p$ is viewed as a random vector with a prior distribution. The algorithm selects a subset $S_t = \phi(\hat p_t,D)$ based on a sample $\hat p_t$  from the prior distribution of $p$ at $t=1$ (or the posterior  distribution at $t\geq 2$), then updates the posterior distribution of $p$ by observations $\{X_{t,i} \}_{i\in S_t}$. For more details, we refer the reader to \citep{russo2017tutorial}.

In our experiment, we consider a residential demand response program with 3000 customers. Each customer can either participate in the DR event by reducing 1kW  or not.  The probabilities of participation are i.i.d. $\text{Unif}[0,1]$. \textcolor{black}{The demand response events last for one hour on each day from June to September in 2018,} with a goal of shaving the 
 peak loads in Rhode Island. The \textit{hourly} demand profile is from New England ISO.\footnote{\url{https://www.iso-ne.com/isoexpress/web/reports/load-and-demand/-/tree/demand-by-zone}} We consider two schemes to determine the peak-load-shaving target $D_t$:
\begin{enumerate}
	\item[i)] Average peak: Compute the averaged  load profile in a day by averaging the  daily load profiles in the four months. The constant target $D$ is the 5\% of the difference between the peak load and the  load at one hour before the peak hour of the averaged load profile. 
	\item[ii)] Daily peak: On each day $t$, the target $D_t$ is 5\% of the difference between the peak load and the load at one hour before the peak hour of the daily demand.
\end{enumerate}

 In our algorithms, we set $\alpha=2.5$. In Thompson sampling, $p$'s prior distribution  is   $\text{Unif}[0,1]^n$. We consider one DR event per day and plot the daily performance. 

Figure \ref{fig: regret} plots the regret of CUCB, CUCB-Avg and TS under the two schemes of peak shaving. The x-axis is in log scale and the resolution is by day. Both figures show that CUCB-Avg  performs  better than  CUCB and  TS. In addition, the regret of CUCB-Avg in Figure \ref{fig: regret}(a) is linear with respect to $\log(T)$,  consistent with our theoretical result in Theorem \ref{thm: regret bdd D}. Moreover, the regret of CUCB-Avg in Figure \ref{fig: regret}(b) is almost linear with $\log(T)$, demonstrating that in practice the regret can be much better than our worst case regret bound in Theorem \ref{thm: Dt regret bound}. 

Figure \ref{fig: conf interval} plots the 90\% confidence interval of the relative  reduction error, $\frac{\sum_{i\in S_t}X_{t,i}-D_t}{D_t}$, of CUCB-Avg and TS by 1000 simulations. It is observed that the relative error of CUCB-Avg roughly stays within $\pm 5\%$,  much better than Thompson sampling. This again demonstrates the reliability of CUCB-Avg. Interestingly, the figure shows that TS tends to reduce less load than the target, which is possibly because TS overestimates the customers' load reduction when selecting customers.  \red{Finally, on August 18th both algorithms cannot fulfill the daily peak target because  it is very hot and the target is too high to reach even after selecting all the users. }

Finally, we compare TS and CUCB-Avg for different $n$ by considering the scheme (i). We consider two cases: 1) when $T$ is very large so the  regret is dominated by the $\log(T)$ term, 2) when $T$ is a reasonable number in practice. We let $T=10^5$ for case 1 and $T=122$ (the total number of days from June to September) for case 2. We consider a smaller target $D=40$ for illustration and consider $n=500:500:3500$. Figure \ref{fig: compare TS with diff. n}(a) shows that the dependence on $n$ of CUCB-Avg's regret is similar to  that of TS when $T$ is large, and the dependence is not cubic,  the theoretical explanation of which is left for future work. Moreover, Figure \ref{fig: compare TS with diff. n}(a) shows that CUCB-Avg can achieve better regrets than TS under a properly chosen small $\alpha$. Though not explained by theory yet,  the phenomenon that a small $\alpha$ yields good performance has been observed in literature \citep{wang2018thompson}. Further, Figure \ref{fig: compare TS with diff. n}(b) shows that CUCB-Avg achieves significantly smaller regrets than TS for a practical $T$, indicating the  effectiveness of our algorithm in reality.

\subsection{More discussion on the effect of $\alpha$ and $n$}
{Figure \ref{fig: compare TS with diff. n} has  shown that the choice of $\alpha$ and $n$ affects the algorithm performance. In this subsection, we will discuss  the effect of $\alpha$ and $n$ in greater details. In particular, we will study the DR performance by the relative deviation of the load reduction, which is defined as $\sqrt{\E L(S_t)}/D_t$, for each day during the four months.  }

Figure \ref{fig:alpha} shows the relative deviation of CUCB-Avg for different $\alpha$ when $n=3000$ and when the target is determined by scheme (i) in Section \ref{subsec: numerical Dt}. It is observed that when $T$ is small, a large $\alpha$ provides smaller relative deviation, thus better performance. This is because   the information of customers is limited when $T$ is small, and larger $\alpha$ encourages exploration of the information, thus yielding better performance. When $T$ is large, a smaller $\alpha$ leads to a better performance. This is because when $T$ is large, the information of customers is sufficient, and a small $\alpha$ encourages the exploitation of the current information, thus generating better decisions. {The observations above are  also consistent with Figure \ref{fig: compare TS with diff. n}}. Further, Figure \ref{fig:alpha} shows that for a wide range of $\alpha$'s values, CUCB-Avg  reduces the deviation to below 5\% after a few days, indicating that CUCB-Avg is reasonably robust to the choice of $\alpha$. 

Figure \ref{fig: n} shows the relative deviation of CUCB-Avg for different $n$ when $\alpha=2.5$ and when the target is determined by scheme (i) in Section \ref{subsec: numerical Dt}. It is observed that even with a large number of customers, CUCB-Avg reduces the relative deviation to below 5\% very quickly, demonstrating that our algorithm can handle large $n$ effectively. In addition, when $T$ is small, a small $n$ provides smaller relative deviation, because a small number of customers is easier to learn in a short time period. When $T$ is large, a large $n$ provides better performance, because there are more reliable  customers to choose from a larger customer pool. {It is worth mentioning that though Figure \ref{fig: compare TS with diff. n}(a) shows that the regret increases with  $n$ when $T$ is large, there is no conflict because the regret captures the gap between the  deviation generated by the algorithm and  the optimal one, which may increase even when the algorithm generates less deviation since the optimal deviation also decreases. }

\begin{figure}
	\centering
		\includegraphics[width=0.3\textwidth]{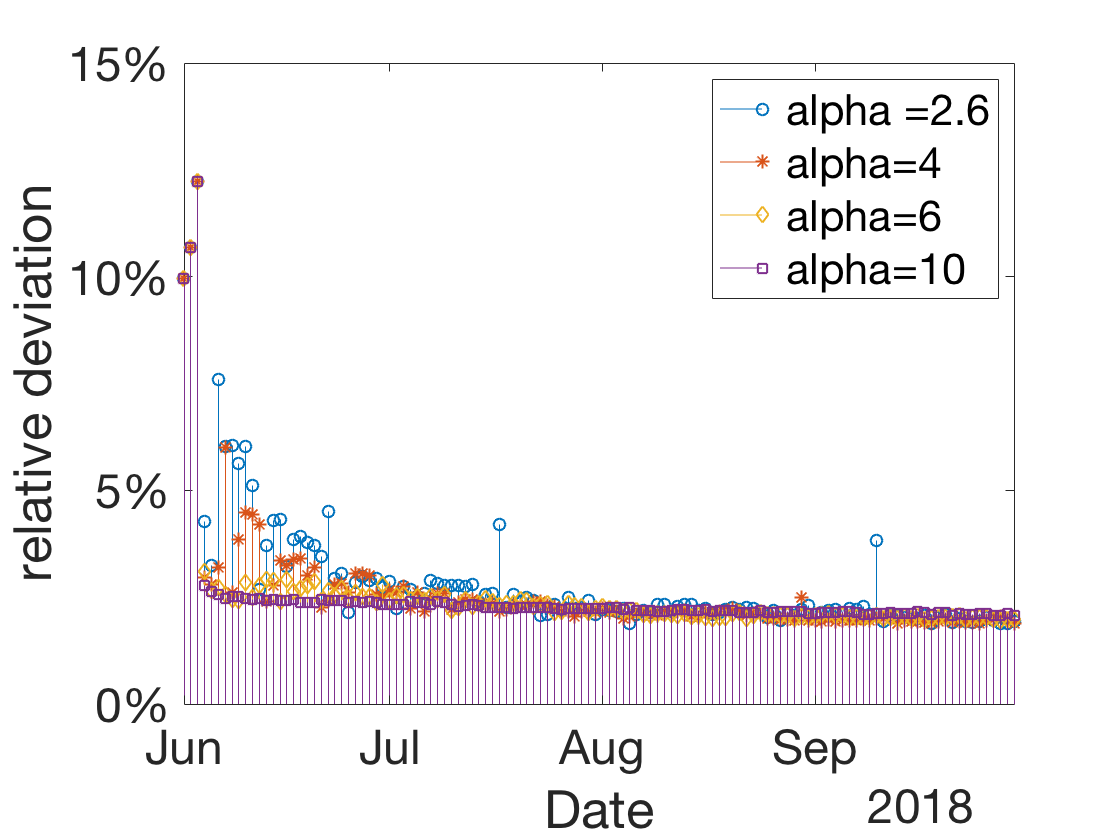}
	\caption{Effect of $\alpha$}
		\label{fig:alpha}
\end{figure}

\begin{figure}
	\centering
	\includegraphics[width=0.3\textwidth]{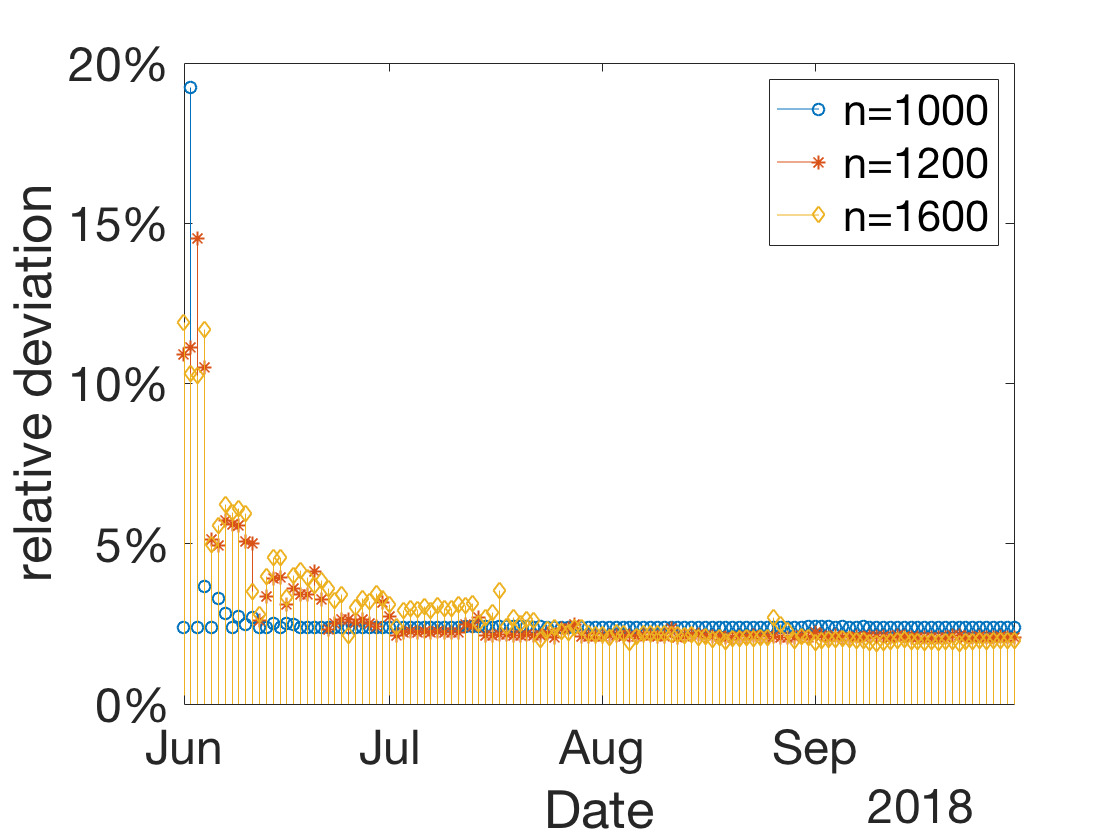}
		\caption{Effect of $n$}
	\label{fig: n}
\end{figure}

\subsection{On the user fatigue effect}

It is widely observed that customers tend to be less responsive to demand response signals after participating in DR events consecutively. This effect is usually called \textit{user fatigue}. Though our algorithm and theoretical analysis do not consider this effect for simplicity, our CUCB-Avg can handle the fatigue effect after small modifications, which is briefly discussed below.

For illustration purpose, we consider a simple model of user fatigue effect.  Each customer $i$ is associated with an original response probability $p_i$. The response probability at stage $t$, denoted as $p_i(t)$, decays exponentially with a fatigue ratio $f_i$ if customer $i$ has been selected consecutively, that is,  $p_i(t)=(f_i)^{\chi_i(t)} p_i$ if customer $i$ has been selected from day $t-\chi_i(t)$ to day $t-1$. If the customer is not selected , we consider that the customer takes a rest at this stage and will respond to the next DR event with the original probability. Though the fatigue model may be too pessimistic about the effects of the consecutive selections by considering exponential decaying fatigue factors,  and too optimistic about the relaxation effect by assuming full recovery after one day rest, this model captures the commonly observed phenomena that the consecutive selection is  a key reason for user fatigue and customers can recover from fatigue if not selected for some time \citep{hopkins2017best}. The model can be revised to be more complicated and realistic, which is left as future work.

Next, we explain how to modify CUCB-Avg to address the user fatigue effects. We consider that the aggregator has some initial estimation of the fatigue ratio of customer $i$, denoted as $\tilde f_i$, and will use the estimated fatigue ratios to rescale the upper confidence bounds and sample averages in Algorithm 2 to account for the fatigue effect. In particular, the rescaled upper confidence bound is $(\tilde f_i)^{\chi_i(t)}U_i(t)$, and the rescaled history sample average is $\bar p_i(t)= \frac{1}{T_i(t)}\sum_{\tau \in I_i(t)} \frac{X_{\tau, i}}{(\tilde f_i)^{\chi_i(t)}}$, where $\chi_i(t)$ denotes the number of consecutive days up until $t-1$ when customer $i$ is selected.

In our numerical experiments, different users may have different user fatigue ratios, which  are generated i.i.d. from $\text{Unif}[0.75, 0.95]$. 
Other parameters are  the same as  in Section 6.1. Figure \ref{fig: consumer fatigue} plots the relative deviation of our modified CUCB-Avg in two scenarios: i) the aggregator  has access to the accurate fatigue ratio, i.e. $\tilde f_i=f_i$; ii) the aggregator only has a rough estimation for the entire population: $\tilde f_i =0.85$ for all $i$. It can be observed that our algorithm is able to reduce the relative deviation to below 5\% after a few days even when the fatigue ratios are inaccurate. This demonstrates that our algorithm, with some simple modifications, can work reasonably well even when considering customer fatigue effects. 

\begin{figure}[htp!]
	\centering
	\includegraphics[width=0.3\textwidth]{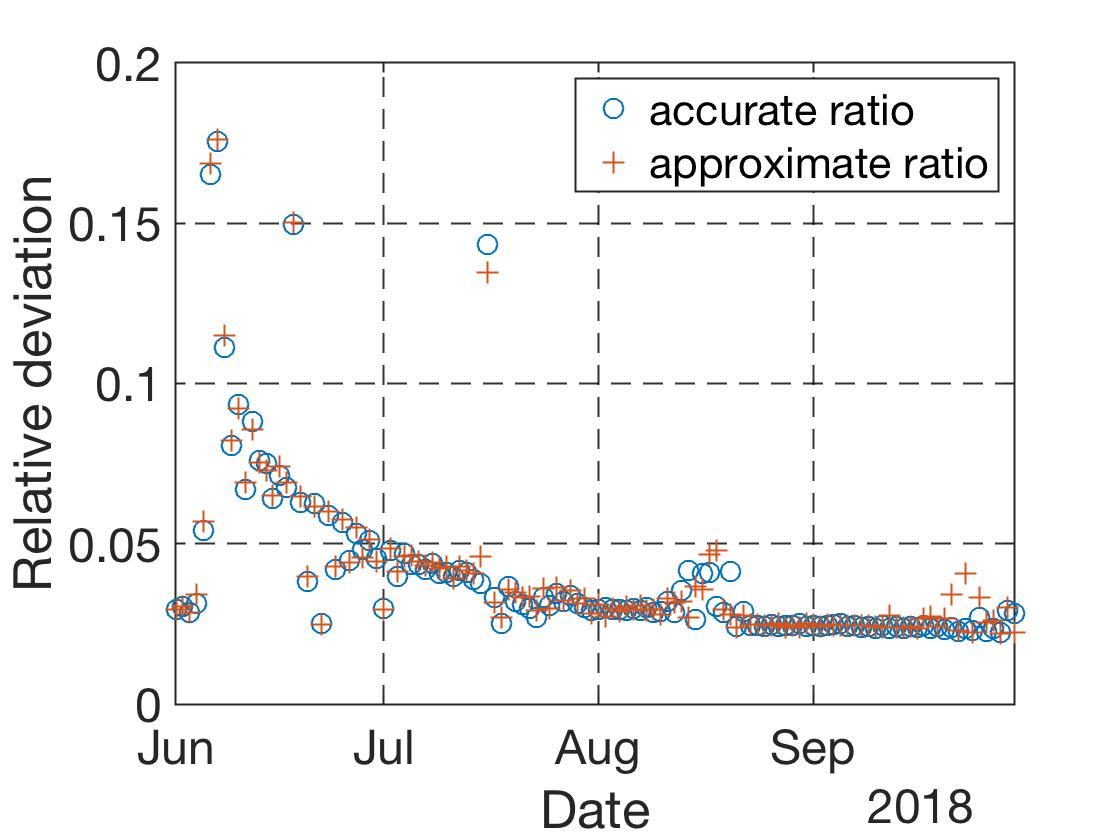}
	\caption{The performance of our CUCB-Avg (after simple modifications) when considering user fatigue.}
	\label{fig: consumer fatigue}
\end{figure}

\section{Conclusion}
This paper studies  a CMAB problem  motivated by residential demand response with the goal of minimizing the difference between the total load adjustment and the target value. 
 We propose  CUCB-Avg and show that CUCB-Avg achieves sublinear regrets in both static and   time-varying cases. 
 There are several interesting directions to explore in the future. First, it is  interesting to improve the dependence on $n$. Second, it is worth  studying the regret lower bounds.   Besides, it is worth considering more realistic behavior models which may include e.g. the effects of temperatures and humidities, the user fatigue, correlation among users,  time-varying response patterns, general load reduction distributions, dynamic population, etc. 

\vspace{12pt}

\appendix

\nbf{Appendix.}

\section{Proof of Theorem \ref{thm: offline opt optimality}}\label{aped: offline opt}

Only in this subsection, we assume $p_1\geq \dots \geq p_n$ to simplify the notation. This will not cause any loss of generality in the offline analysis.  In other parts of the document, the order of $p_1,\dots, p_n$ is unknown, and we will use $p_{\sigma(1)}\geq \dots p_{\sigma(n)}$ to denote the non-increasing order of parameters.

Since adding or removing an arm $i$ with $p_i=0$ will not affect the regret, in the following we will assume $p_i>0$ without loss of generality \footnote{One way to think about this is that we only consider subset $S$ such that $p_i>0$ for $i\in S$.}.



The proof of Theorem \ref{thm: offline opt optimality} takes two steps. 
\begin{enumerate}
	\item In Lemma \ref{lem:necessary condition}, we establish a local optimality condition, which is a necessary condition to the global optimality. 
	\item In Lemma \ref{lem: exists S*}, we show that there exists an optimal selection $S^*$ which includes first $k$ arms with an unknown $k$. Then we can easily show that the Algorithm \ref{alg: offline oracle} selects the optimal $k$.
\end{enumerate}

We first state and prove Lemma \ref{lem:necessary condition}. 

\begin{lemma}
	\label{lem:necessary condition}
	Suppose $S^*$ is optimal and $p_i >0$ for $i\in S^*$. Then we must have 
	$$\sum_{i\in S^*-\{j\} }p_i\leq D-1/2,\quad \forall \ j \in S^*$$
	If $S^*\not =[n]$, then we will also have
	\begin{align*}
	& \sum_{i\in S^*} p_i\geq D-1/2
	\end{align*}
	
\end{lemma}

\begin{proof}
	Since $S^*$ is optimal, removing an element will not reduce the expected loss, i.e. 
	$$\E L(S^*)\leq \E L(S^*-\{j\})$$
	which is equivalent with
	\begin{align*}
	&\E L(S^*)- \E L(S^*-\{j\})\\
	& = (2\sum_{i\in S^*} p_i -2D-p_j)p_j + p_j(1-p_j)\\
	& = (2\sum_{i\in S^*} p_i -2D+1-2p_j)p_j \leq 0
	\end{align*}
	Since $p_j>0$, we must have
	$$\sum_{i\in S^*-\{j\}} p_i\leq D-1/2$$

	If $S^*\not =[n]$, then adding an element will not reduce the cost. So we must have
	\begin{align*}
	&	\E L(S^*)\leq \E L(S^*\cup \{j\}), \quad \forall j\not\in S^*
	\end{align*}
	which is equivalent with
	\begin{align*}
	&\E L(S^*)- \E L(S^*\cup\{j\})\\
	& =- (2\sum_{i\in S^*} p_i -2D+p_j)p_j - p_j(1-p_j)\\
	& = -(2\sum_{i\in S^*} p_i -2D+1)p_j \leq 0
	\end{align*}
	Since $p_j>0$, we must have
	$$\sum_{i\in S^*} p_i\geq D-1/2.$$

\end{proof}

\begin{proof}[Proof of Corollary \ref{cor: D<1/2}]
	Suppose there exists a non-empty optimal subset $S\not=\emptyset$. Then 
	$$\sum_{i\in S^*-\{j\} }p_i\geq 0 > D-1/2,\quad \forall \ j \in S^*$$
	which results in a contradiction by Lemma \ref{lem:necessary condition}.
\end{proof}

\begin{corollary}\label{cor: small n}
	When $\sum_{i=1}^n p_i \leq D-1/2$, the optimal subset is  $[n]$.
\end{corollary}
\begin{proof}
	Suppose there exists an optimal subset $S\not=[n]$. Then 
	$$\sum_{i\in S^* }p_i< \sum_{i=1}^n p_i \leq D-1/2$$
	which results in a contradiction by Lemma \ref{lem:necessary condition}.
\end{proof}

Next, we are going to show that there must exist an optimal subset  containing all elements with highest mean values. This is done by contradiction.
\begin{lemma}\label{lem: exists S*}
	When $D\geq1/2$ and $\sum_{i=1}^n p_i> D-1/2$,
	there must exist an optimal subset $S^*$ whose elements' mean values are $q_1\geq \dots\geq q_m$, such that for any $p_i>q_m$, we have $i \in S^*$.
\end{lemma}
\begin{proof}
	Let's prove by construction and contradiction. 
	Consider $S^*$, assume there exists $p_i> q_m$ but $i\not \in S^*$. In the following, we will ignore other random variables outside $S^*\cup \{i\}$ because they are irrelevant. Now, we rank the mean value $\{q_1, \dots,p_i,\dots ,q_m, \}$ and assume that $p_i$ is the $j$th largest element here. To simplify the notation, we will  call the newly ranked mean value set as
	$$ p_1\geq \dots \geq p_{j-1} \geq p_j\geq p_{j+1}\geq \dots \geq p_{m+1}$$
	The mean values of random variables in $S^*$ are $\{p_1,\dots ,p_{j-1} , p_{j+1}, \dots, p_{m+1}\}$ (used to be called as $\{q_1,\dots, q_m\}$) and the injected element (used to be denoted as $p_i$) now is called $p_j$. Under this simpler notation, we proceed to construct a  subset of top $k$ arms with some $k$ whose expected loss is no more than the optimal expected loss.
	
	Construct a  subset $A$ in the following way. Pick the smallest $k$ such that 
	\begin{align*}
	& \sum_{i=1}^k p_i> D-1/2\\
	& \sum_{i=1}^{k-1} p_i\leq D-1/2
	\end{align*}
	Then let $A=\{1,\dots, k\}$. It is easy to see that $k\geq j$. ( Since $S^*$ is optimal,  by Lemma \ref{lem:necessary condition}, excluding any element will go below $D-1/2$, so $k$ must include the newcomer $j$ to be beyond $D-1/2$.)
	
	We claim that $\E L(A)\leq \E L(S^*)$. Since $\E L(S^*)$ is optimal, we must have $\E L(A)=\E L(S^*)$ and $A$ is also optimal. Then, we can construct a new subset $A_1$ with the same rule above. Since there are only finite elements, we can always end up with an optimal set $\hat A$ which includes variables with the highest mean values. Then the proof is done.
	
	The only remaining part is to prove $\E L(A)\leq \E L(S^*)$. Denote 
	\begin{align*}
	& \sum_{i=1}^k p_i= D-1/2+\delta_x\\
	& \sum_{i\not=j,\, i=1}^{m+1} p_i=D-1/2+\delta_y
	\end{align*}
	By construction of $k$, $\delta_x\in(0, p_k]$. By Lemma \ref{lem:necessary condition}, $\delta_y\in [0, p_{m+1}]$.
	So we have
	\begin{align*}
	\sum_{i=1}^k p_i-\sum_{i\not=j,\, i=1}^{m+1} p_i=p_j-\sum_{i=k+1}^{m+1} p_i=\delta_x-\delta_y
	\end{align*}
	
	Now let's do some simple algebra and try to prove $\E L(A)-\E L(S^*)\leq 0$. Basically, we are trying to write  $\E L(A)-\E L(S^*)$ with $\delta_x$ $\delta_y$ defined above and then bound it using bounds of  $\delta_x$ $\delta_y$ .
	\begin{align}
	&\E L(A)-\E L(S^*)=(\sum_{i=1}^k p_i -D)^2 + \nonumber\\
	&\sum_{i=1}^k p_i(1-p_i)-(\sum_{i\not=j,\, i=1}^{m+1} p_i -D)^2 -\sum_{i\not=j,\, i=1}^{m+1} p_i(1-p_i)\nonumber\\
	&=(\sum_{i=1}^k p_i -D+\sum_{i\not=j,\, i=1}^{m+1} p_i-D )\nonumber\\
	&	(\sum_{i=1}^k p_i -D-\sum_{i\not=j,\, i=1}^{m+1} p_i +D)\nonumber\\
	&+p_j(1-p_j)-\sum_{i=k+1}^{m+1} p_i(1-p_i)\nonumber\\
	&=(\delta_x+\delta_y-1)(\delta_x-\delta_y)+p_j-\sum_{i=k+1}^{m+1} p_i+\sum_{i=k+1}^{m+1} p_i^2-p_j^2\nonumber\\
	&=(\delta_x+\delta_y-1)(\delta_x-\delta_y)+\delta_x-\delta_y+\sum_{i=k+1}^{m+1} p_i^2-p_j^2\nonumber\\
	&=(\delta_x+\delta_y)(\delta_x-\delta_y)+\sum_{i=k+1}^{m+1} p_i^2-p_j^2\nonumber\\
	&=(\delta_x+\delta_y)(\delta_x-\delta_y)+\sum_{i=k+1}^{m+1} p_i^2-(\sum_{i=k+1}^{m+1} p_i+\delta_x-\delta_y)^2\nonumber\\
	&=(\delta_x+\delta_y)(\delta_x-\delta_y)+\sum_{i=k+1}^{m+1} p_i^2 \nonumber\\
	&-\left[ (\sum_{i=k+1}^{m+1} p_i)^2 + (\delta_x-\delta_y)^2 +2(
	\delta_x-\delta_y)\sum_{i=k+1}^{m+1} p_i\right]\nonumber\\
	&=(\delta_x-\delta_y)(\delta_x+\delta_y-\delta_x+\delta_y-2\sum_{i=k+1}^{m+1} p_i)\nonumber\\
	&+\sum_{i=k+1}^{m+1} p_i^2-(\sum_{i=k+1}^{m+1} p_i)^2\nonumber\\
	&=(\delta_x-\delta_y)(2\delta_y-2\sum_{i=k+1}^{m+1} p_i)+\sum_{i=k+1}^{m+1} p_i^2-(\sum_{i=k+1}^{m+1} p_i)^2 \label{equ::1}
	\end{align}
	
	Now, we first notice that $\delta_y\leq p_{m+1}\leq \sum_{i=k+1}^{m+1} p_i$, so $(2\delta_y-2\sum_{i=k+1}^{m+1} p_i)\leq 0$. 
	
	Also notice that $$\sum_{i=k+1}^{m+1} p_i^2-(\sum_{i=k+1}^{m+1} p_i)^2\leq 0$$ since $p_i\geq 0$.
	
	\noindent\textbf{Case 1:} $\delta_x\geq \delta_y$. In this case, \eqref{equ::1}$\leq 0$ is straightforward.

	\noindent\textbf{Case 2:} $\delta_x< \delta_y$. In this case, $p_{m+1}<p_j<\sum_{i=k+1}^{m+1} p_i$. So we must have $m-k\geq 1$. Since $(2\delta_y-2\sum_{i=k+1}^{m+1} p_i)\leq 0$, we can decrease $\delta_x$ to 0
	\begin{align*}
	&	\eqref{equ::1}\leq -\delta_y(2\delta_y-2\sum_{i=k+1}^{m+1} p_i)+\sum_{i=k+1}^{m+1} p_i^2-(\sum_{i=k+1}^{m+1} p_i)^2\\
	&=RHS
	\end{align*}
	RHS is a quadratic function with respect to $\delta_y$ and it is increasing in the region $[0, \frac{\sum_{i=k+1}^{m+1} p_i}{2}]$. 
	Since $m-k\geq 1$, we have
	$$ \frac{\sum_{i=k+1}^{m+1} p_i}{2}\geq (p_{m+1}+p_m)/2\geq p_{m+1}\geq \delta_y$$
	So the highest possible value is reached when $\delta_y=p_{m+1}$. Plugging this in RHS, we have
	\begin{align*}
	&RHS\\
	&\leq -p_{m+1}(2p_{m+1}-2\sum_{i=k+1}^{m+1} p_i)+\sum_{i=k+1}^{m+1} p_i^2-(\sum_{i=k+1}^{m+1} p_i)^2\\
	&=-2p_{m+1}^2 +2(\sum_{i=k+1}^{m+1} p_i)p_{m+1}+\sum_{i=k+1}^{m+1} p_i^2-(\sum_{i=k+1}^{m+1} p_i)^2\\
	&=(\sum_{i=k+1}^{m+1} p_i)(p_{m+1}-\sum_{i=k+1}^{m} p_i)-p_{m+1}^2+\sum_{i=k}^{m} p_i^2\\
	&=(p_{m+1}+\sum_{i=k}^{m} p_i)(p_{m+1}-\sum_{i=k+1}^{m} p_i)-p_{m+1}^2+\sum_{i=k}^{m} p_i^2\\
	&=p_{m+1}^2-(\sum_{i=k}^{m} p_i)^2-p_{m+1}^2+\sum_{i=k}^{m} p_i^2\\
	&=\sum_{i=k}^{m} p_i^2-(\sum_{i=k}^{m} p_i)^2\leq 0 \quad \quad\tag{Since $p_i\geq 0$}
	\end{align*}
	Thus we have shown that 
	$$ \E L(A)-\E L(S^*)\leq 0$$

\end{proof}

Now, given Lemma \ref{lem:necessary condition} and Lemma \ref{lem: exists S*}, we can prove Theorem \ref{thm: offline opt optimality}.

\textit{Proof of theorem \ref{thm: offline opt optimality}:}
When $D<1/2$ or $\sum_{i=1}^n p_i\leq D-1/2$, see Corollary \ref{cor: D<1/2} and Corollary \ref{cor: small n}. 

When $D\geq 1/2$ and $\sum_{i=1}^n p_i> D-1/2$,
by Lemma \ref{lem: exists S*}, there exists an optimal subset $S^*=\{1,\dots, k\}$ for some $k$, i.e. containing the first several arms with largest mean values. Since $S^*$ is optimal, we must have, by Lemma \ref{lem:necessary condition}, that 
\begin{align*}
& \sum_{i=1}^k p_i\geq  D-1/2\\
& \sum_{i=1}^{k-1} p_i\leq D-1/2
\end{align*}

\begin{enumerate}
	\item[If] $\sum_{i=1}^k p_i> D-1/2$, then  $S^*$ is the output of Algorithm \ref{alg: offline oracle}, so the output of Algorithm \ref{alg: offline oracle} is optimal.
	\item[If] $\sum_{i=1}^k p_i=  D-1/2$, then it is easy to show that $\E L(\{1,\dots, k\})=\E L(\{1,\dots, k+1\})$. So $\{1,\dots, k+1\}$ is also optimal. The output of Algorithm \ref{alg: offline oracle} is $\{1,\dots, k+1\}$. So the output of Algorithm \ref{alg: offline oracle} is still optimal.
\end{enumerate}
\qed
\begin{corollary}\label{cor: case 2}
	If $\sum_{i=1}^{k-1} p_i=  D-1/2$, then the subset $\{1,\dots, k-1\}$ and  $\{1,\dots, k\}$ are both optimal.
\end{corollary}


\section{Proof of Proposition \ref{prop: A1A2 epsilon0} and Lemma  \ref{lem: no regret bar E bar B} with Assumption (A1) and (A2)}\label{aped: D no regret Ass}
We will prove Lemma \ref{lem: no regret bar E bar B} by using $\epsilon_0$'s expression in Proposition \ref{prop: A1A2 epsilon0} under Assumption (A1) and (A2). The proof of Proposition \ref{prop: A1A2 epsilon0} follows naturally. 

Let $\F_t$ denote the natural filtration up to time $t$.

Before the proof, we provide two technical  lemmas that characterize the properties of $S_t$ selected by CUCB-Avg, which will be useful not only in this section but also in the sections afterwards.
\begin{lemma}[Properties of CUCB-Avg's Selection]\label{lem:alg properties}
	For any $t$ that is not in the initialization phase, the subset $S_t$ selected by CUCB-Avg satisfies the following properties
	\begin{enumerate}
		\item[i)] $\sum_{i\in S_t} \bar p_{i}(t-1) >D-1/2 $ or $S_t=[n]$.
		\item[ii)] Define $i_{\min}\in\argmin_{i\in S_t}U_i(t)$, then $$\sum_{i\in S_t-\{i_{\min}\}}\bar p_{i}(t-1) \leq D-1/2. $$
		\item[iii)] For any $j\in S_t$ and any $i\in [n]$ such that $U_i(t)> U_j(t)$, we have $i\in S_t$.
	\end{enumerate}
	
\end{lemma}

\begin{proof}
The proof is straightforward by Algorithm \ref{alg: ucb_avg}.
\end{proof}

\begin{lemma}\label{lem: epsilon < Deltak}
	For any $\epsilon>0$ satisfying $\epsilon\leq\Delta_k/2$ for some $k$ with $0\leq k \leq n$,  given $\bar E_t$ and $\bar B_t(\epsilon)$,  we have either $S_t \subseteq \{ \sigma(1),\dots, \sigma(k)  \}$ or $\{ \sigma(1),\dots, \sigma(k)  \} \subseteq S_t$.
\end{lemma}

\begin{proof}
	When $k=0$ or $n$, the statement is trivially true.
	
	When $1\leq k \leq n-1$.
	Suppose there is a realization of $\F_{t-1}$ such that $\bar E_t$ and $\bar B_t(\epsilon)$ hold, but $S_t \not \subseteq \{ \sigma(1),\dots, \sigma(k)  \}$ and $\{ \sigma(1),\dots, \sigma(k)  \}\not \subseteq S_t$. Fix this realization of $\F_{t-1}$, then there exists $j \in S_t - \{ \sigma(1),\dots, \sigma(k)  \}$ and $i\in \{ \sigma(1),\dots, \sigma(k)  \}- S_t$. We will show that $U_i(t) >U_j(t)$ in the following
	\begin{align*}
	U_j(t) &\leq \bar p_j(t-1) +\sqrt{\frac{\alpha \log t}{2T_j(t-1)}} \\
	&< p_j + 2\sqrt{\frac{\alpha \log t}{2T_j(t-1)}}  < p_j + 2\epsilon   \\
	& \leq p_{\sigma(k+1)}+2\epsilon  \leq p_{\sigma(k+1)}+2\frac{p_{\sigma(k)}-p_{\sigma(k+1)}}{2}  \\
	& = p_{\sigma(k)} \leq p_i \\
	& \leq \min\left( \bar p_i(t-1) + \sqrt{ \frac{\alpha \log t}{2 T_i(t-1)} },1\right) = U_i(t)
	\end{align*}
	where the first inequality and last equality are from the definition of the upper confidence bound \eqref{eq:Ui}, the second inequality uses the fact that $\bar E_t$ holds, the third inequality is based on $\bar B_t(\epsilon)$, the fourth inequality is by our choice of $j$, the fifth inequality is by $\epsilon\leq \Delta_k/2$, the sixth inequality is by our choice of $i$, and the last inequality uses the fact that $\bar E_t$ and $\bar B_t(\epsilon)$ hold.

	Together with Lemma \ref{lem:alg properties} (iii), we have shown that $i \in S_t$, which leads to a contradiction.
\end{proof}

\nit{Proof of Lemma \ref{lem: no regret bar E bar B} given Assumption (A1) and (A2):} 

Notice that $S^*=\phi(p, D)=\{\sigma(1),\dots, \sigma(k)\}$. By Lemma \ref{lem: epsilon < Deltak} and $\epsilon_0 \leq \Delta_k/2$, we have either $S_t \subseteq S^*$ or $S^* \subseteq S_t$ given $\bar E_t$ and $\bar B_t(\epsilon_0)$. In the following, we will first show that $S_t=S^*$ in Step 1-2  then prove zero regret in Step 3.


\nbf{Step 1: Given $\bar E_t,\bar B_t(\epsilon_0)$, $S_t\subsetneqq S^*$ is impossible: } We prove this by contradiction.
Suppose $S_t\subsetneqq S^*$, then $S_t \not = [n]$ and
\begin{align*}
&	\sum_{i\in S_t }\bar p_i(t-1) \\
&<  \sum_{i\in S_t}\left(  p_i+ \sqrt{ \frac{\alpha \log t}{2 T_i(t-1)} } \right) \qquad \tag{by $\bar E_t$}\\
& < \sum_{i\in S_t}\left(  p_i+ \epsilon_0\right) \qquad \tag{by $\bar B_t(\epsilon_0)$}\\
& \leq \sum_{i=1}^{k-1}\left( p_{\sigma(i)}+ \epsilon_0\right) \qquad \tag{by $S_t\subsetneqq S^*$ }\\
& = D-1/2-\delta_2 +(k-1)\epsilon_0  \qquad \tag{by definition of $\delta_2$}\\
& \leq D-1/2 \qquad \tag{by definition of $\epsilon_0$, $\epsilon_0 \leq \frac{\delta_2}{k}$}
\end{align*}
However, by Lemma \ref{lem:alg properties} (i), $ \sum_{i\in S_t }\bar p_i(t-1)> D-1/2$, which leads to a contradiction.

\nbf{Step 2: Given $\bar E_t,\bar B_t(\epsilon_0)$, $S^*\subsetneqq S_t$ is impossible: }We prove this by contradiction.
Suppose $S^*\subsetneqq S_t$, so $S^*\not =[n]$, thus $\sum_{i=1}^n p_i>D-1/2$.  We denote $i_{\min}\in \argmin_{i\in S_t} U_i(t)$. We will first  show that $i_{\min} \in S_t-S^*$, then show that $\sum_{i\in S_t-\{ i_{\min}\}} \bar p_i(t-1) > D-1/2$. Thus, by Lemma \ref{lem:alg properties} (ii), we have a contradiction. 


Now, first of all, we show that $i_{min}\in  S_t-S^*$. It suffices to show that for any $i\in S^*$, and $j\in S_t-S^*$,  we have $U_i(t)>U_j(t)$. This is proved by the following.
\begin{align*}
U_j(t) & < p_j + 2\epsilon_0  \qquad \tag{by $\bar E_t$ and $\bar B_t(\epsilon_0)$}\\
& \leq p_{\sigma(k+1)} + 2 \frac{\Delta_k}{2} \quad \tag{by $j \not \in S^*$ and def. of $\epsilon_0$}\\
& = p_{\sigma(k)}  \leq p_i \leq U_i(t)  \tag{by $\bar E_t$ and $\bar B_t(\epsilon_0)$}
\end{align*}

Then we show $\sum_{i\in S_t-\{ i_{\min}\}} \bar p_i(t-1) > D-1/2$ by
\begin{align*}
\sum_{i\in S_t-\{ i_{min}\}} \bar p_i(t-1) & \geq \sum_{i \in S^*} \bar p_i(t-1)\\
&> \sum_{i\in S^*}  (p_i-\epsilon_0)  \\
&= D-1/2+\delta_1 -k\epsilon_0 \\
& \geq D-1/2 
\end{align*}
where the first inequality is by $S_t-\{i_{\min}\} \supseteq S_t$, and the  second inequality is by $\bar E_t, \bar B_t(\epsilon_0)$.

\nbf{Step 4: Prove $\E R_t(S_t)I_{\{ \bar E_t,\bar B_t(\epsilon_0) \}} =0$. }
By the three steps above, we have $S_t=S^*$ under $\bar E_t,\bar B_t(\epsilon_0) $, then it is straightforward that
\begin{align*}
&\E [R_t(S_t)I_{\{ \bar E_t,\bar B_t(\epsilon_0) \}} ]
= \E [R_t(S_t)I_{\{ \bar E_t,\bar B_t(\epsilon_0),S_t = S^*\}}]\\
& = \E [R_t(S^*)I_{\{ \bar E_t,\bar B_t(\epsilon_0),S_t = S^*\}}]=0
\end{align*}
\qed


\section{A general expression of $\epsilon_0$ and a proof of Lemma  \ref{lem: no regret bar E bar B} without Assumption (A1) and (A2)}\label{aped: D no regret}
Without Assumption (A1) and (A2), we need additional technical discussion because there  might be multiple optimal subsets due to ties in the probability profile $p$ and Corollary \ref{cor: case 2}. But the main idea behind the proof is the same.

We will first give an explicit expression of $\epsilon_0$, then we will show $\E R_t(S_t)I_{\{\bar E_t, \bar B_t(\epsilon_0) \}}=0$ given the new $\epsilon_0$.

Without loss of generality, we will consider $D\geq 1/2$ due to Corollary \ref{cor: D<1/2}.

We denote the natural filtration up to time $t$ as $\F_t$.

Now we present the expression of $\epsilon_0$ in the general case.

\begin{definition}\label{def: general epsilon0}
	Let $p_{\sigma(1)}\geq \dots \geq p_{\sigma(n)}$. The $\epsilon_0$ in Lemma \ref{lem: no regret bar E bar B} can be determined by
	$$ \epsilon_0=\min(\frac{\delta_1}{l_1} , \frac{\delta_2}{l_2} , \frac{\Delta_{k_1}}{2}, \frac{\Delta_{k_2}}{2})$$
where  parameter $\delta_1, \delta_2, l_1, l_2, k_1, k_2$ are defined below.
	
	\nbf{Definition of $l_1$:} 
	\begin{align*}
	l_1 =
	\begin{cases}
	\min G_1& \text{if } \sum_{i=1}^n p_i >D-1/2\\
	n & \text{ if } \sum_{i=1}^n p_i \leq D-1/2
	\end{cases}
	\end{align*}
	where $$G_1 = \{1\leq k \leq n | \ \sum_{i=1}^k p_{\sigma(i)} >D-1/2 \}. $$ 
	Notice that $\{\sigma(1),\dots, \sigma(l_1)  \}$ is one possible output of  Algorithm \ref{alg: offline oracle} (there might be other outputs due to the random tie-breaking rule). 
	
	
	
	\nbf{Definition of $\delta_1$:} 
	\begin{align*}
	\delta_1 =
	\begin{cases}
	\sum_{i=1}^{l_1} p_{\sigma(i)} -(D-1/2)& \text{if } \sum_{i=1}^n p_i>D-1/2\\
	n & \text{ if } \sum_{i=1}^n p_i\leq D-1/2
	\end{cases}
	\end{align*}
	Note that when $l_1=n$, we let $\delta_1 =n$, so that $\frac{\delta_1}{l_1}=1$ is large enough to not affect the value of $\epsilon_0$.
	
	\nbf{Definition of $l_2$:} 
	\begin{align*}
	l_2 =
	\max  \{0\leq k <l_1 | \ \sum_{i=1}^k p_{\sigma(i)} <D-1/2 \}
	\end{align*}
	 Note that $0\leq l_2\leq n-1$, and $l_1> l_2$.

	\nbf{Definition of $\delta_2$:} 
	\begin{align*}
	\delta_2=
	\begin{cases}
	(D-1/2) - 	\sum_{i=1}^{l_2} p_{\sigma(i)}& \text{if } l_2\geq 1\\
	1 & \text{ if } l_2=0
	\end{cases}
	\end{align*}
	Note that when $l_2=0$, let $\delta_2=1$, then $\frac{\delta_2}{l_2}=+\infty$, which is large enough to not affect the value of $\epsilon_0$

	\nbf{Definition of $\Delta_i$:} 
	\begin{align*}
	\Delta_i  & =
	\begin{cases}
	p_{\sigma(i)} - p_{\sigma(i+1)}  & \text{ if } 1\leq i \leq n-1\\
	2 & \text{ if } i=0, n
	\end{cases}
	\end{align*}
	Note that when $i=0$ or $n$, we let $\Delta_i= 2$ which is large enough to keep the $\epsilon_0$ unaffected.
	
	\nbf{Definition of $k_1$:} 
	\begin{align*}
	k_1 = \max \{ 0\leq  i \leq l_1-1| \  \Delta_i>0\}
	\end{align*}
Intuitively, $k_1$ is the  largest index under the non-increasing order such that the parameter $p_{\sigma(k_1)} $ is the second smallest among $p_{\sigma(1)}, \dots, p_{\sigma(l_1)}$.
	
	\nbf{Definition of $k_2$:} 
	\begin{align*}
	k_2 = \min \{ l_1\leq  i \leq n| \  \Delta_i>0\}
	\end{align*}
	Intuitively, $k_2$ is the largest index under the non-increasing order that can be possibly selected by the offline optimization algorithm. 
\end{definition}

In the following, we will first prove a supportive corollary based on Lemma \ref{lem: epsilon < Deltak}, then provide a proof of Lemma  \ref{lem: no regret bar E bar B}.

\begin{corollary}\label{cor: D divide S*}
	Given $\bar E_t$ and $\bar B_t(\epsilon_0)$, then we have  $S_t \subsetneqq \{ \sigma(1),\dots, \sigma(k_1)  \}$, or $\{ \sigma(1),\dots, \sigma(k_1)  \} \subseteq S_t \subseteq \{ \sigma(1),\dots, \sigma(k_2)  \}$, or $\{ \sigma(1),\dots, \sigma(k_2)  \} \subsetneqq S_t$.
\end{corollary}

\begin{proof}
	It is easy to see that $k_1 < k_2$ by definition. By using the fact that $\epsilon_0\leq \Delta_{k_1}/2$, $\epsilon_0\leq \Delta_{k_1}/2$, and  Lemma \ref{lem: epsilon < Deltak}, it is straightforward to prove the corollary.
\end{proof}

Finally, we are ready to prove Lemma  \ref{lem: no regret bar E bar B}.

\nit{Proof of Lemma  \ref{lem: no regret bar E bar B}:} The major part of the proof is to show that if $\bar E_t,\bar B_t(\epsilon_0)$ happen, then $S_t$ must be optimal. Then, given that $S_t$ is optimal, it is easy to prove zero regret at time $t$. 

Now, let's first prove $S_t\in \Ss^*$ given $\bar E_t,\bar B_t(\epsilon_0)$, where $\Ss^*$ denotes the set of all possible optimal subsets.

\nbf{Step 1: prove that $ \ S_t \in \Ss^*$ when $\bar E_t$ and $\bar B(\epsilon_0) $ hold.} We will list all possible scenarios,  and prove that in each scenario,  when $\bar E_t$ and $\bar B(\epsilon_0)$ hold, we have $S_t \in \Ss^*$.

\underline{Scenario 1: When $l_1 >l_2+2$,} we will first show that $\Ss^*$ contains any set $S$ that satisfies $S\supseteq \{\sigma(1),\dots, \sigma(l_2+1)  \}$. To prove this, we first mention that 
   the following facts can be verified based on the definitions: $k_2=l_1=n$, $\sum_{i=1}^{l_2+1 }p_{\sigma(i)} = \sum_{i=1}^{l_2+2 }p_{\sigma(i)}=D-1/2$, $\sum_{i=1}^n p_i=D-1/2$, $p_{\sigma(l_2+1)}>0$, $p_{\sigma(l_2+2)}= \dots =p_{\sigma(n)}=0$, $k_1=l_2+1$.
Moreover, by Corollary \ref{cor: case 2},  $\{\sigma(1),\dots, \sigma(l_2+1)  \}$ is optimal. Since  the union of set  $\{\sigma(1),\dots, \sigma(l_2+1)  \}$ and some arms with zero probabilities is also optimal, any set with a subset $\{\sigma(1),\dots, \sigma(l_2+1)  \}$ is also optimal. 


Next, we will show $S_t$ is optimal. By Corollary \ref{cor: D divide S*} and $k_1=l_2+1$ and $k_2=n$, we have either $S_t \subsetneqq \{ \sigma(1),\dots, \sigma(l_2+1)  \}$ or $\{ \sigma(1),\dots, \sigma(l_2+1)  \} \subseteq S_t$. Since the second possible case guarantees $S_t \in \Ss^*$, we only need to show
 that $S_t \subsetneqq \{ \sigma(1),\dots, \sigma(l_2+1)  \}$ is impossible. This is done by contradiction. Suppose $S_t \subsetneqq \{ \sigma(1),\dots, \sigma(l_2+1)  \}$, then
\begin{align*}
&\sum_{i\in S_t}\bar p_{i}(t-1) < \sum_{i=1}^{l_2}\bar p_{\sigma(i)}(t-1)\\
&\leq \sum_{i=1}^{l_2} p_{\sigma(i)}+\epsilon_0 \tag{by $\bar E_t$, $\bar B_t(\epsilon_0)$}\\
& =  \sum_{i=1}^{l_2} p_{\sigma(i)}+l_2\epsilon_0\leq D-1/2  \tag{by def. of $l_2$ and $\epsilon_0$}
\end{align*} 
By Lemma \ref{lem:alg properties} (i), this leads to a contradiction. Therefore, $S_t$ is optimal in this scenario.

\underline{Scenario 2: When $l_1 =l_2+1$,} we will first show that $\Ss^*$ contains any set that contains and only contains $\sigma(1),\dots, \sigma(k_1) $ together with $l_1-k_1$ arms from subset $\{\sigma(k_1+1),\dots, \sigma(k_2) \}$. To prove this, notice that by definition of $l_1$ and Theorem \ref{thm: offline opt optimality}, $\phi(p, D)=\{\sigma(1), \dots, \sigma(l_1) \} $ is optimal. Then, by definition of $k_1$ and $ k_2$, we have $k_1+1 \leq l_1 \leq k_2$ and the arms $\{\sigma(k_1+1),\dots, \sigma(k_2) \}$ are in a tie with the same parameter $p_{\sigma(l_1)}$. Moreover, there are $l_1-k_1$ arms in $\phi(p, D)$ whose value is $p_{\sigma(l_1)}$. Therefore,  replacing these $l_1-k_1$ arms with any $l_1-k_1$ arms with the same value will still yield an optimal subset. This proves that any set is optimal if it contains $\sigma(1),\dots, \sigma(k_1) $ together with $l_1-k_1$ arms from subset $\{\sigma(k_1+1),\dots, \sigma(k_2) \}$.

Next, we will show $S_t$ is optimal. By Corollary \ref{cor: D divide S*}, $S_t$ satisfies one of the three possibilities (a) $S_t \subsetneqq \{ \sigma(1),\dots, \sigma(k_1)  \}$, or (b) $\{ \sigma(1),\dots, \sigma(k_1)  \} \subseteq S_t \subseteq \{ \sigma(1),\dots, \sigma(k_2)  \}$, or  (c) $\{ \sigma(1),\dots, \sigma(k_2)  \} \subsetneqq S_t$. We will show that (a) and (c) are impossible. Then, we will show that $S_t$ has exactly $l_1-k_1$ arms with parameter  $p_{\sigma(k_1+1)}$. Thus, $S_t$ is optimal. 


 Firstly, we suppose the possibility (a) is true, i.e. $S_t \subseteq  \{\sigma(1),\dots, \sigma(k_1)  \}$, then we have
\begin{align*}
&\sum_{i\in S_t}\bar p_{i}(t-1)\leq \sum_{i\in S_t}\left( p_{i}+\epsilon_0\right) \tag{by $\bar E_t, \bar B_t(\epsilon_0)$}\\
& \leq  \sum_{i=1}^{k_1} p_{\sigma(i)}+k_1\epsilon_0 \\
& \leq  \sum_{i=1}^{l_2} p_{\sigma(i)}+l_2\epsilon_0 \tag{by def. $k_1 \leq l_1-1=l_2$}\\
&\leq D-1/2 \tag{by def. of $l_2, \delta_2, \epsilon_0$}
\end{align*}
which contradicts Lemma \ref{lem:alg properties} (i). Hence, we have $S_t \supsetneqq \{\sigma(1),\dots, \sigma(k_1)  \}$.

Secondly, we suppose possibility (c) is true, i.e. $S_t \supsetneqq  \{\sigma(1),\dots, \sigma(k_2)  \} $.  We denote $i_{min}\in \argmin_{i\in S_t} U_i(t)$. It can be shown that $i_{min} \in S_t- \{\sigma(1),\dots, \sigma(k_2)  \} $. This is because for any $i\in  \{\sigma(1),\dots, \sigma(k_2)  \} $ and any $j\in S_t- \{\sigma(1),\dots, \sigma(k_2)  \} $, we have
\begin{align*}
U_j(t) & < p_j + 2\epsilon_0  & \tag{by $\bar E_t$ and $\bar B_t(\epsilon_0)$}\\
& \leq p_{\sigma(k_2+1)} + 2 \frac{\Delta_{k_2}}{2} \\
& = p_{\sigma(k_2)}\leq p_i \leq U_i(t)
\end{align*}
where the first and last inequality use the definition of $U_i(t)$ in \eqref{eq:Ui} and the fact that $\bar E_t$ and $\bar B_t(\epsilon_0)$ are true. As a result, $i_{min}\not\in \{\sigma(1),\dots, \sigma(k_2)  \} $. Therefore, we have
\begin{align*}
&\sum_{i\in S_t-\{ i_{min}\}} \bar p_i(t-1) \geq \sum_{i=1}^{k_2}\bar p_{\sigma(i)}(t-1)  \\
& \geq  \sum_{i=1}^{l_1}\bar p_{\sigma(i)}(t-1)   > \sum_{i=1}^{l_1} p_{\sigma(i)} -l_1\epsilon_0 \\
& =D-1/2+\delta_1 -l_1\epsilon_0 \geq D-1/2 
\end{align*}
where the last equality uses the definition of $\delta_1$ and the fact that $\sum_{i=1}^n p_i>D-1/2$ when $S_t \supsetneqq \{\sigma(1), \dots, \sigma(k_2)\}$. 
This leads to a contradiction with Lemma \ref{lem:alg properties} (ii). Thus, (c) is not true.

%
%

Consequently, we have  $\{ \sigma(1),\dots, \sigma(k_1)  \} \subseteq S_t \subseteq \{ \sigma(1),\dots, \sigma(k_2)  \}$. In the following, we  will show that $S_t$ has exactly $l_1-k_1$ arms with parameter $p_{\sigma(k_1+1)}$ by contradiction. By using the same proof techniques as above, it is straightforward to show that if we select more or less than $l_1-k_1$ arms, the sum of parameters is either more than $D-1/2$ excluding $i_{min}$ or less than $D-1/2$, which leads to a contradiction with Lemma \ref{lem:alg properties} (ii) and (iii). 

In conclusion, $S_t$ is optimal in this scenario.

\underline{Scenario 3: When $l_1 =l_2+2$,} it can be shown that $\Ss^*$ contains subsets that contains and only contains $\sigma(1),\dots, \sigma(k_1) $ together with either $l_1-k_1$ or $l_1-1-k_1$ arms with parameter $p_{\sigma(l_1)}$. To prove this, notice that by $l_1 =l_2+2$, we have $\sum_{i=1}^{l_1-1}p_{\sigma(i)}=D-1/2$. Hence by Corollary \ref{cor: case 2}, both $\{\sigma(1),\dots, \sigma(l_1)  \}$ and $ \{\sigma(1),\dots, \sigma(l_1-1)  \}$ are optimal. In addition, by definition of $k_1$ and $ k_2$, we have $k_1+1 \leq l_1 \leq k_2$ and the arms $\{\sigma(k_1+1),\dots, \sigma(k_2) \}$ are in a tie with the same parameter $p_{\sigma(l_1)}$. Moreover, there are $l_1-k_1$ and $l_1-k_1-1$ arms with parameter $p_{\sigma(l_1)}$ in $\{\sigma(1),\dots, \sigma(l_1)  \}$ and $ \{\sigma(1),\dots, \sigma(l_1-1)  \}$ respectively. Therefore,  replacing these  arms by  arms with the same parameter will still yield an optimal subset. This proves that any set is optimal if it contains $\sigma(1),\dots, \sigma(k_1) $ together with $l_1-k_1$ or $l_1-k_1-1$ arms from subset $\{\sigma(k_1+1),\dots, \sigma(k_2) \}$.


Next, we can prove that $S_t$ is optimal in the same way as in Scenario 2. By Corollary \ref{cor: D divide S*}, $S_t$ satisfies one of the three possibilities (a) $S_t \subsetneqq \{ \sigma(1),\dots, \sigma(k_1)  \}$, or (b) $\{ \sigma(1),\dots, \sigma(k_1)  \} \subseteq S_t \subseteq \{ \sigma(1),\dots, \sigma(k_2)  \}$, or  (c) $\{ \sigma(1),\dots, \sigma(k_2)  \} \subsetneqq S_t$. We can show that (a) and (c) are impossible. Then, we can show that $S_t$ has either $l_1-k_1$ or $l_1-1-k_1$ arms with parameter  $p_{\sigma(l_1)}$. The proof is the same as that for Scenario 2 above, thus being omitted here for brevity. 
%
%
%
%


\nbf{Step 2: prove $\E R_t(S_t)I_{\{ \bar E_t,\bar B_t(\epsilon_0) \}} =0$. } Notice that
$\bar E_t$ and $\bar B_t(\epsilon_0)$ are determined by $\F_{t-1}$. Consider $\Ss^* $ to be the set of all optimal subsets, 
	\begin{align*}
	& \E I_{\{\bar E_t,\bar B_t(\epsilon_0)\}} R_t(S_t)=
\E I_{\{\bar E_t,\bar B_t(\epsilon_0), \ S_t\in \Ss^*\}} R_t(S_t) \\
	& =  \sum_{S\in \Ss^*} \E I_{\{\bar E_t,\bar B_t(\epsilon_0), \ S_t=S\}} R_t(S_t)  \\
	& =  \sum_{S\in \Ss^*} \Pb(\bar E_t,\bar B_t(\epsilon_0), \ S_t=S )\E R_t(S)  =0
	\end{align*}
	where the second inequality is because $R_t(S) $ and $ I_{\{\bar E_t,\bar B_t(\epsilon_0), \ S_t=S\}} $ are independent. 

\qed

\section{Proof of Lemma \ref{lem: Sst under assumption Dt}}\label{append: regret error if epsilon1>xi with assumption}

To illustrate the intuition, we first provide the proof  under the assumption that $p$ and $D_t$ satisfies Assumption (A1) and (A2) in this appendix. Then, we will provide a proof without these assumptions based on the same idea in Appendix \ref{append: regret error if epsilon1>xi}.

Suppose $p$ and $D_t$ satisfies Assumption (A1) and (A2),
we define    constants: $\epsilon_1$, $\xi(D_t)$ and $\epsilon_0(D_t)$  as

$$ \epsilon_1= \min\left(\frac{\Delta_1}{2},\dots, \frac{\Delta_T}{2}, \frac{\beta}{n}\right)$$

\begin{equation}
\xi(D_t)= 
\min\left(\frac{\delta^t_1}{l^t_1}, \frac{\delta^t_2}{l^t_1}, \epsilon_1\right) 
\end{equation}
$$ \epsilon_0(D_t)= \min \left(\frac{\delta^t_1}{l^t_1}, \frac{\delta^t_2}{l^t_1},\frac{\Delta_{l_1^t}}{2}\right)$$
where
\begin{align*}
& l_1^t=| \phi(p, D_t)|\\
& \sum_{i=1}^{l_1^t} p_{\sigma(i)}=D_t-1/2+\delta^t_1, \\
& \sum_{i=1}^{l_1^t-1} p_{\sigma(i)}=D_t-1/2-\delta^t_2.
\end{align*}
Notice that the constant $\epsilon_0(D_t)$ is the same as the constant $\epsilon_0$ defined in Proposition \ref{prop: A1A2 epsilon0} with respect to  target $D_t$.


\begin{proof}
	We are going to discuss two different scenarios based on different values of $\epsilon$ and prove the bound in each scenario.
	
	\nbf{Scenario 1: when $\xi(D_t)\geq \epsilon$, show zero regret.}  In this case, we have $\epsilon\leq \xi(D_t)\leq \epsilon_0(D_t)$. As a result, we have $\bar E_t \cap \bar B_t(\epsilon)\subseteq \bar E_t \cap \bar B_t(\epsilon_0(D_t))$, thus, by Lemma \ref{lem: no regret bar E bar B}, there is no regret at $t$. 
	
	\nbf{Scenario 2: when $\xi(D_t)< \epsilon$, show $S_t$ at most differs from the optimal set by one arm. }
	Formally, we will show that, conditioning on $\F_{t-1}$ such that $\bar E_t$ and $\bar B_t(\epsilon)$ hold, the selection of CUCB-Avg must satisfy
	$S_t \in \{ \{  \sigma(1), \dots, \sigma(l_1^t-1)   \} , \{ \sigma(1), \dots, \sigma(l_1^t)  \}, \{ \sigma(1), \dots, \sigma(l_1^t+1) \}\}$. The proof takes three steps.
	
	Step 1: $S_t$ must satisfy one of the five conditions: i) $S_t \subsetneqq \{  \sigma(1), \dots, \sigma(l_1^t-1)   \}$, ii) $S_t= \{  \sigma(1), \dots, \sigma(l_1^t-1)   \}$, iii), $S_t = \{  \sigma(1), \dots, \sigma(l_1^t)   \}$, iv) $S_t =\{  \sigma(1), \dots, \sigma(l_1^t+1)   \}$, v) $S_t \supsetneqq \{  \sigma(1), \dots, \sigma(l_1^t+1)   \}$. This is proved by   $\epsilon_1< \min\left(\frac{\Delta_{l_1^t-1}}{2}, \frac{\Delta_{l_1^t}}{2},\frac{\Delta_{l_1^t+1}}{2}\right)$ and Lemma \ref{lem: epsilon < Deltak} in Appendix \ref{aped: D no regret Ass}.
	
	
	Step 2: Show that condition i) is not possible by contradiction. We suppose $S_t \subsetneqq \{  \sigma(1), \dots, \sigma(l_1^t-1)   \}$, and show that the total estimated mean is less than $D_t-1/2$ below, which contradicts  Lemma \ref{lem:alg properties} (i).
	\begin{align*}
	&	\sum_{i\in S_t }\bar p_i(t-1) \\
	&<  \sum_{i\in S_t}\left(  p_i+ \sqrt{ \frac{\alpha \log t}{2 T_i(t-1)} } \right) \qquad \tag{by $\bar E_t$}\\
	& < \sum_{i\in S_t}\left(  p_i+ \epsilon_1\right) \qquad \tag{by $\bar B_t(\epsilon_1)$}\\
	& \leq \sum_{i=1}^{l_1^t-2}\left( p_{\sigma(i)}+ \epsilon_1\right) \ \tag{by $S_t \subsetneqq \{  \sigma(1), \dots, \sigma(l_1^t-1)   \}$ }\\
	& = D_t-1/2-\delta^t_2 -p_{\sigma(l_1^t-1)}+(l_1^t-2)\epsilon_1  \qquad \tag{by def of $\delta^t_2$}\\
	& \leq D_t-1/2 \qquad \tag{by  $\epsilon_1 \leq \frac{\beta}{n}\leq \frac{p_{\sigma(l_1^t-2)}}{l_1^t-2}$}
	\end{align*}
	
	Step 3: Show that condition v) is not possible. This is proved by contradiction. Suppose v) is true, then it can be shown that $i_{\min}=\argmin_{i\in S_t} U_i(t)$ must be in the set $S_t- \{  \sigma(1), \dots, \sigma(l_1^t+1)   \}$ based on the same argument in the Step 2 of the proof of Lemma \ref{lem: no regret bar E bar B} in Appendix \ref{aped: D no regret Ass}. In addition, based on the same argument, we can show that $\sum_{i\in S_t-\{ i_{min}\}}\bar p_i(t-1)>D_t/2$ below:
	\begin{align*}
	\sum_{i\in S_t -\{i_{\min}\}} \bar p_i(t-1)
	&\geq	\sum_{i=1}^{l_1^t+1}\bar p_{\sigma(i)}(t-1) \\
	&>  \sum_{i=1}^{l_1^t+1}\left(  p_{\sigma(i)}- \sqrt{ \frac{\alpha \log t}{2 T_i(t-1)} } \right)\\
	& > \sum_{i=1}^{l_1^t+1}\left(  p_{\sigma(i)}- \epsilon_1\right) \\
	& = D_t-1/2+\delta^t_1+p_{\sigma(l_1^t+1)}-(l_1^t+1)\epsilon_1 \\
	& \geq  D_t-1/2
	\end{align*}
	where the second and third  inequality are based on $\bar E_t$ and $\bar B_t(\epsilon_1)$, and the last equality and inequality are based on the definition of $\delta_1^t$ and $\epsilon_1$.

	By Lemma \ref{lem:alg properties} (ii), there is a contradiction. Therefore, v) is not true.

	\nbf{Scenario 2 (continued): when $\xi(D_t)< \epsilon$, show $\E[ R_t(S_t)\mid \F_{t-1}] \leq 2n\epsilon$.}
	We only need to discuss condition ii) and iv) since the regret is zero in  condition iii).
	
	Conditioning on condition ii), we have
	\begin{align*}
	&\E [R_t(S_t)\mid \F_{t-1}, S_t=\{\sigma(1),\dots,\sigma(l_1^t-1)\} ]\\
	& = (\sum_{i=1}^{l_1^t-1}p_{\sigma(i)}-D_t)^2 -  (\sum_{i=1}^{l_1^t}p_{\sigma(i)}-D_t)^2 - p_{\sigma(l_1^t)}(1-p_{\sigma(l_1^t)})\\
	&=-p_{\sigma(l_1^t)}(2\sum_{i=1}^{l_1^t-1}p_{\sigma(i)}-2D_t+p_{\sigma(l_1^t)})- p_{\sigma(l_1^t)}(1-p_{\sigma(l_1^t)})\\
	&= -p_{\sigma(l_1^t)}(2\sum_{i=1}^{l_1^t-1}p_{\sigma(i)}-2D_t+1)\\
	&= 2p_{\sigma(l_1^t)}\delta_{2}^t\leq 2\delta_{2}^t \leq 2n \epsilon
	\end{align*}
	where the last inequality uses the fact that $\delta_2^t < (l_1^t-1)\epsilon$ when $S_t=\{\sigma(1),\dots,\sigma(l_1^t-1)\}$ and $\bar E_t$, $\bar B_t(\epsilon)$ hold, which is proved below.
	\begin{align*}
	D_t-1/2&<  \sum_{i\in S_t} \bar p_i(t-1)=	\sum_{i=1}^{l_1^t-1}\bar p_{\sigma(i)}(t-1) \\
	& \leq \sum_{i=1}^{l_1^t-1}\left( p_{\sigma(i)}+ \epsilon\right)  \ \tag{by def of $\bar E_t$ and $\bar B_t(\epsilon)$}\\
	& = D_t-1/2-\delta_2^t +(l_1^t-1)\epsilon
	\end{align*}
	
	
	Conditioning on condition iv), we have
	\begin{align*}
	&	\E [R_t(S_t)\mid \F_{t-1}, S_t=\{\sigma(1),\dots, \sigma(l_1^t+1)\}]\\
	& = (\sum_{i=1}^{l_1^t+1}p_{\sigma(i)}-D_t)^2 -  (\sum_{i=1}^{l_1^t}p_{\sigma(i)}-D_t)^2\\
	&\quad +p_{\sigma(l_1^t+1)}(1-p_{\sigma(l_1^t+1)})\\
	&=p_{\sigma(l_1^t+1)}(2\sum_{i=1}^{l_1^t}p_{\sigma(i)}-2D_t+p_{\sigma(l_1^t+1)})\\
	&\quad+ p_{\sigma(l_1^t+1)}(1-p_{\sigma(l_1^t+1)})\\
	&= p_{\sigma(l_1^t+1)}(2\sum_{i=1}^{l_1^t}p_{\sigma(i)}-2D_t+1)\\
	&= 2p_{\sigma(l_1^t)}\delta_{1}^t\leq 2\delta_{1}^t \leq 2n \epsilon
	\end{align*}
	where the last inequality uses  the fact that $\delta_t^1 < l_1^t\epsilon$ when $S_t=\{\sigma(1),\dots, \sigma(l_1^t+1)\}$ and $\bar E_t$, $\bar B_t(\epsilon)$ hold,  which  is proved  below.
	\begin{align*}
	D_t-1/2&\geq \sum_{i\in S_t-\{i_{\min}\}} \bar p_i(t-1)\\
	& >  \sum_{i\in S_t-\{i_{\min}\}} \left(  p_i- \sqrt{ \frac{\alpha \log t}{2 T_i(t-1)} } \right)\\
	&=  \sum_{i=1}^{l_1^t}\left(  p_{\sigma(i)}- \sqrt{ \frac{\alpha \log t}{2 T_i(t-1)} } \right) \\
	& > \sum_{i=1}^{l_1^t}\left(  p_{\sigma(i)}- \epsilon\right) \\
	& = D_t-1/2+\delta^t_1-l_1^t\epsilon_1 
	\end{align*}
	where the first inequality is by Lemma \ref{lem:alg properties} (ii), the second inequality is by $\bar E_t$, the last inequality is by $\bar B_t(\epsilon)$, the first equality uses the fact that $\min_{i\in S_t} U_i(t)= U_{\sigma(l_1^t+1)}(t)$ due to $\epsilon\leq \epsilon_1 \leq \Delta_{l_1^t}/2$ and $\bar E_t$ and $\bar B_t(\epsilon)$.
	
	In conclusion, we have $\E[R_t(S_t)\mid \F_{t-1}]\leq 2n\epsilon$.
\end{proof}


\subsection{Proof of Lemma \ref{lem: Sst under assumption Dt} without additional assumptions}\label{append: regret error if epsilon1>xi}
In this appendix, we provide a proof of Lemma \ref{lem: Sst under assumption Dt} without additional assumption (A1) (A2). 

Before the proof, we note that adding or deleting a zero-valued arm from the selected subset $S_t$ will not affect the regret $\E [R_t(S_t)\mid \F_{t-1}]$. Therefore, without loss of generality, we will focus on $S_t$ without zero-valued arms.

In addition, we note that zero regret under $\xi(D_t)> \epsilon_1$ can be proved in the same way as Lemma \ref{lem: no regret bar E bar B}. Therefore, we only need to focus on $\xi(D_t) \leq \epsilon_1$.

\nbf{Scenario 1: target $D_t$ is reachable.} In this scenario, there exists $1\leq l_1^t\leq n$ such that 
$$ \sum_{i=1}^{l_1^t} p_{\sigma(i)}>D_t-1/2$$
$$ \sum_{i=1}^{l_2^t} p_{\sigma(i)}< D_t-1/2$$
This also suggests that $p_{\sigma(l_1^t)}>0$.

We will characterize $S_t$ in two ways. First, we will show that $S_t$ must follow the right ordering. Second, we will show that $S_t$ will at most select one more or one less arm than the optimal selection from the oracle.

Firstly, 
by Lemma \ref{lem: epsilon < Deltak} and $\epsilon_1 \leq \min\left( \Delta_{k_1^t}, \Delta_{k_2^t}\right)$, we know the subset $S_t$ selected by CUCB-Avg must satisfy one of the following:
\begin{enumerate}
	\item[i)]  $S_t \subseteq \{\sigma(1), \dots, \sigma(k_1^t) \}$, 
	\item[ii)] $\{ \sigma(1),\dots, \sigma(k_1^t)  \} \subsetneqq S_t \subseteq \{ \sigma(1),\dots, \sigma(k_2^t)  \}$, and $S_t$ has no more than $l_2^t- k_1^t$ arms whose parameter is equal to $p_{\sigma(k_1^t)}$. 
	\item[iii)] $\{ \sigma(1),\dots, \sigma(k_1^t)  \} \subsetneqq S_t \subseteq \{ \sigma(1),\dots, \sigma(k_2^t)  \}$, and $S_t$ has  $l_2^t-k_1^t+1\leq u \leq l_1^t-k_1^t$ arms whose parameter is equal to $p_{\sigma(k_1^t)}$. 
	\item[iv)]  $\{ \sigma(1),\dots, \sigma(k_1^t)  \} \subsetneqq S_t \subseteq \{ \sigma(1),\dots, \sigma(k_2^t)  \}$, and $S_t$ has more than $l_1^t- k_1^t$ arms whose parameter is equal to $p_{\sigma(k_1^t)}$. 
	\item[v)] $\{ \sigma(1),\dots, \sigma(k_2^t)  \} \subsetneqq S_t$.
\end{enumerate} 

By the proof of Lemma \ref{lem: no regret bar E bar B}, iii) generates no regret. Similarly to the proof in Appendix \ref{append: regret error if epsilon1>xi with assumption}, we can rule out i) and v) and show that ii) and iv) only happens under some restrictions on $\delta_1^t$ and $\delta_2^t$. Then following the same argument as in the proof with Assumption (A1) and (A2), we can bound the regret by $2n\epsilon$.

\nbf{Scenario 2: target $D_t$ too large to reach.} In this scenario, $\sum_{i=1}^n p_i \leq D_t-1/2$. Similarly, we can show that $S_t$ satisfy $|S_t|\geq n-1$ and $S_t$ must include the top $n-1$ arms. Only when $\delta_2^t <(n-1)\epsilon$, the regret is not 0, and the regret bound will also hold by the same argument.
\section{Incorporating the ideas of risk-aversion MAB}\label{append: risk averse}

\begin{figure}
	\centering
	\includegraphics[width=0.3\textwidth]{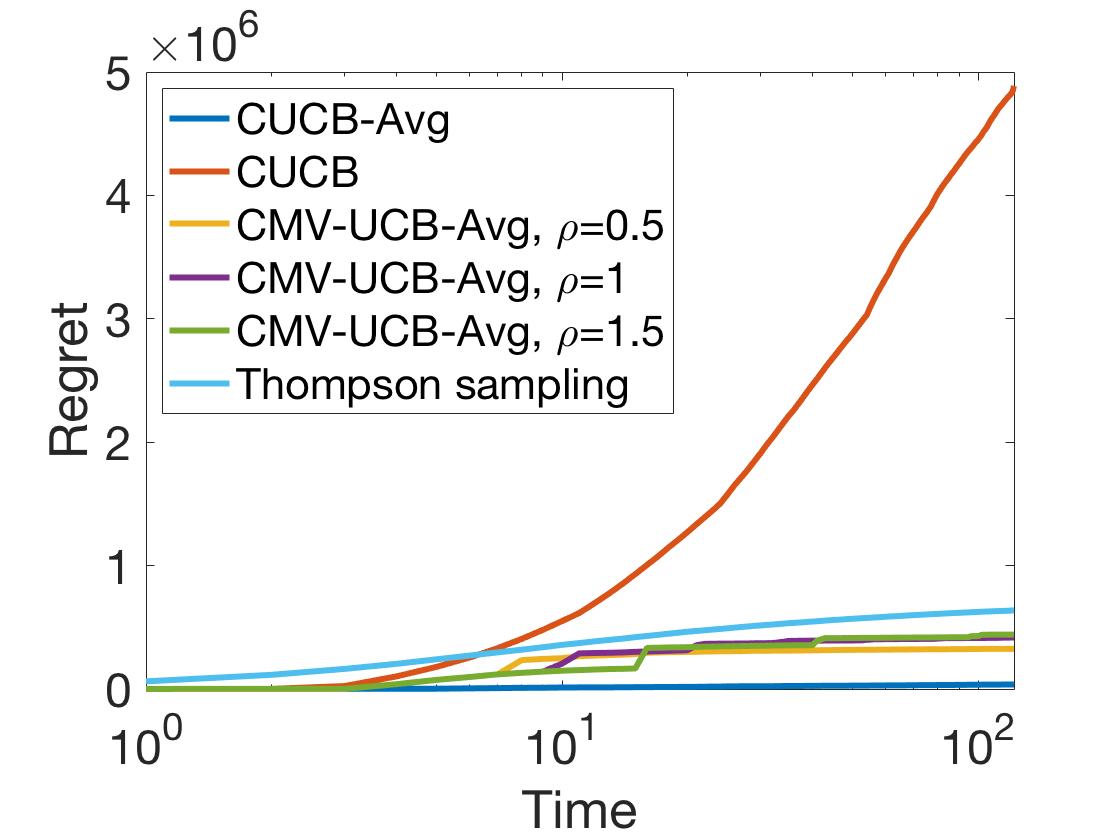}
	\caption{Regret comparison between CUCB, CUCB-Avg, CMV-UCB-Avg, and Thompson sampling}
	\label{fig: mvucb}
\end{figure}

As mentioned in Section 1.2, the papers on risk aversion MAB  \citep{sani2012risk,vakili2016risk} focus on selecting the \textit{single} arm with the best \textit{mean}-variance tradeoff, while our paper aims at selecting \textit{a subset of} arms to achieve the best \textit{bias}-variance tradeoff, where the bias refers to the difference between the expected load reduction and the target load reduction. Identifying the single arm with the best mean-variance tradeoff is helpful, but not enough to  ensure the load reduction to be close to the target. Therefore, the risk-aversion MAB algorithms cannot be directly applied to solve our problem.

Nevertheless, out of curiosity, we  combine the risk-aversion ideas and our algorithm design ideas to construct a new algorithm, which we call CMV-UCB-Avg.  CMV-UCB-Avg ranks arms by the MV-UCB index proposed in \cite{sani2012risk,vakili2016risk}, which consists of an empirical mean-variance tradeoff and an upper confidence bound, then selects the top $K$ arms according to the step 2 in our CUCB-Avg.  The first-rank-then-select structure is motivated by our offline optimization algorithm. We conduct numerical experiments to compare CMV-UCB-Avg with other algorithms under the average-peak setting  in Section 6. Figure \ref{fig: mvucb} shows that CMV-UCB-Avg performs better than the classic CUCB. However, our CUCB-Avg  performs better than CMV-UCB-Avg for several different values of the mean-variance tradeoff parameter $\rho$.



\bibliographystyle{agsm}
\bibliography{citation4MAB}         

\end{document}